\documentclass[11pt,a4paper]{article}
\usepackage[T1]{fontenc}
\usepackage[utf8]{inputenc}
\usepackage{newtxtext,newtxmath}

\usepackage[a4paper,margin=40mm]{geometry}
\usepackage{cmap}
\usepackage{amsfonts}

\usepackage{amsthm}

\usepackage{mathtools}
\allowdisplaybreaks
\usepackage{bm}
\usepackage{graphicx}

\usepackage{microtype}
\usepackage[numbers,sort&compress]{natbib}
\usepackage{xcolor}
\usepackage{booktabs}
\usepackage{enumitem}
\usepackage{array}
\usepackage{tabularx}
\usepackage{setspace}
\usepackage[explicit]{titlesec}
\usepackage[font=small,labelfont=bf]{caption}
\usepackage{mdframed}
\usepackage{needspace}
\usepackage{tocloft}

\newcolumntype{Y}{>{\raggedright\arraybackslash}X}
\AtBeginEnvironment{tabular}{\small}
\AtBeginEnvironment{tabularx}{\small}

\numberwithin{equation}{section}
\mathtoolsset{showonlyrefs}

\setlist{nosep,leftmargin=1.25em}

\titleformat{\section}
{\normalfont\normalsize\bfseries\raggedright}
{\thesection}{0.75em}{\begin{minipage}[t]{0.9\textwidth}#1\end{minipage}}
[\vspace{0.12em}\titlerule]

\newcommand{\subaccent}{\vspace{0.01em}\noindent\rule{7ex}{0.2pt}}
\titleformat{\subsection}
{\normalfont\normalsize\bfseries}{\thesubsection}{0.5em}{#1}
[\subaccent]
\titlespacing*{\subsection}{0pt}{2.2ex plus 0.6ex}{1.5ex}

\titleformat{\subsubsection}
{\normalfont\normalsize\bfseries}{\thesubsubsection}{0.5em}{#1}
\titlespacing*{\subsubsection}{0pt}{1.6ex plus 0.4ex}{0.6ex}

\usepackage[hidelinks,bookmarks=false]{hyperref}

\usepackage[capitalise,noabbrev]{cleveref}
\crefname{equation}{Eq.}{Eqs.}
\crefname{lemma}{Lemma}{Lemmas}
\crefname{theorem}{Theorem}{Theorems}
\crefname{proposition}{Proposition}{Propositions}
\crefname{corollary}{Corollary}{Corollaries}
\crefname{section}{Section}{Sections}
\crefname{appendix}{Appendix}{Appendices}

\newtheoremstyle{thmcompact}
{0.5ex}{0.5ex}
{\itshape}
{0pt}
{\bfseries}
{.}
{0.6em}
{\thmname{#1}\ \thmnumber{#2}\ \thmnote{(\textit{#3})}}

\theoremstyle{thmcompact}
\newtheorem{theorem}{Theorem}[section]
\newtheorem{proposition}[theorem]{Proposition}
\newtheorem{lemma}[theorem]{Lemma}
\newtheorem{corollary}[theorem]{Corollary}

\theoremstyle{remark}
\newtheorem*{remarkplain}{Remark}

\mdfdefinestyle{thmframe}{%
	leftline=true, rightline=false, topline=false, bottomline=false,
	linewidth=0.8pt, linecolor=black,
	innerleftmargin=8pt, innerrightmargin=8pt,
	innertopmargin=6pt, innerbottommargin=6pt,
	skipabove=6pt, skipbelow=6pt,
	nobreak=true
}
\surroundwithmdframed[style=thmframe]{theorem}
\surroundwithmdframed[style=thmframe]{proposition}
\surroundwithmdframed[style=thmframe]{lemma}
\surroundwithmdframed[style=thmframe]{corollary}

\newmdenv[
linewidth=0.6pt,
topline=false,
bottomline=false,
rightline=false,
skipabove=\baselineskip,
skipbelow=\baselineskip,
backgroundcolor=gray!3,
leftmargin=0pt,
innerleftmargin=8pt,
innerrightmargin=6pt,
frametitleaboveskip=0.4em,
frametitlebelowskip=0.2em,
frametitlefont=\bfseries\small,
]{remarkbox}
\usepackage{titling}
\newenvironment{remark}[1][]%
{\begin{remarkbox}\begin{remarkplain}[#1]\normalfont}
		{\end{remarkplain}\end{remarkbox}}

\titleformat{name=\section,numberless}
{\normalfont\normalsize\bfseries\raggedright}
{}{0em}{\begin{minipage}[t]{0.9\textwidth}#1\end{minipage}}
[\vspace{0.12em}\titlerule]
\makeatletter
\renewcommand{\l@section}[2]{\@dottedtocline{1}{0em}{2.3em}{\raggedright #1}{#2}}
\renewcommand{\l@subsection}[2]{\@dottedtocline{2}{1.8em}{3.2em}{\raggedright #1}{#2}}
\makeatother

\newcommand{\thmneed}{\Needspace{6\baselineskip}}
\AtBeginEnvironment{theorem}{\thmneed}
\AtBeginEnvironment{proposition}{\thmneed}
\AtBeginEnvironment{lemma}{\thmneed}
\AtBeginEnvironment{corollary}{\thmneed}

\AtBeginEnvironment{thebibliography}{\small}

\title{\textbf{The Converse Madelung Question} \\
\smallskip
\large{Schrödinger Equation from Minimal Axioms}}
\author{J.\,R.~Dunkley}
\date{31st October 2025}
\begin{document}
\maketitle
	
\begin{abstract}
	We define and address the converse Madelung question: not whether Fisher information can reproduce quantum mechanics, but whether it is necessary. We adopt minimal, physically motivated axioms on hydrodynamic variables: locality, probability conservation, Euclidean invariance with global $\mathrm{U}(1)$ phase symmetry, reversibility, and convex regularity. Within the ensuing explicitly restricted class of first-order local Hamiltonian field theories, the Poisson bracket is uniquely fixed to the canonical bracket on $(\rho,S)$ under the Dubrovin-Novikov hypotheses for local first-order hydrodynamic brackets with probability conservation. Under a pointwise, gauge-covariant complexifier $\psi=\sqrt{\rho}\,e^{iS/\hbar}$, among convex, rotationally invariant, first-derivative local functionals of $\rho$, the only one whose Euler-Lagrange contribution yields a reversible completion that 
	becomes exactly projectively linear is the Fisher functional. With $\hbar^{2}=2m\alpha$ the dynamics reduce to the linear Schrödinger equation. In many-body systems, exact projective linearity with a single local complex structure across all sectors forces $\alpha_i=\hbar^{2}/(2m_i)$ componentwise, thereby fixing a single Planck constant. Galilean covariance appears via the Bargmann central extension in this framework, with the usual superselection implications. Comparison with the Doebner-Goldin family identifies the reversible corner at zero diffusion as the linear Schrödinger case in our variables. We supply operational falsifiers via residual diagnostics for the continuity and Hamilton-Jacobi equations and report numerical minima at the Fisher scale that are invariant under Galilean boosts. These results are consistent with viewing quantum mechanics, in this setting, as a reversible fixed point of Fisher-regularised information hydrodynamics. A code archive accompanies the work for direct numerical verification and reproducibility, including a superposition stress-test showing that, in our tested families and to numerical precision under grid refinement, the Fisher regulariser preserves exact projective linearity within our axioms.
\end{abstract}

\newpage
\tableofcontents
\newpage

\section{Introduction}
Quantum mechanics is typically postulated through the linear Schrödinger equation
\begin{equation}
	i\hbar\,\partial_t\psi=\left[-\frac{\hbar^2}{2m}\nabla^2+V(x,t)\right]\psi,
	\label{eq:LSE}
\end{equation}
Madelung showed in 1927 that writing $\psi=\sqrt{\rho}\,e^{iS/\hbar}$ decomposes Eq.~\eqref{eq:LSE} into a continuity equation and a Hamilton-Jacobi equation regularised by a quantum potential proportional to Fisher information~\cite{madelung1927}. Bohm and Hall-Reginatto later showed that adding a Fisher-information term to a classical ensemble reproduces quantum dynamics~\cite{bohm1952,hall2002}. These results, however, establish only sufficiency: Fisher regularisation can yield the Schrödinger form, but this does not show that it is necessary.

We ask whether Fisher curvature is not only sufficient but necessary within a strictly local, first order, reversible Hamiltonian class on $(\rho,S)$ endowed with Euclidean invariance, global $\mathrm{U}(1)$
 phase symmetry, and convex regularity in $\rho$ alone. Formally: classify all admissible first derivative convex regularisers $f(\rho,\nabla\rho)$ and compatible Poisson brackets on $(\rho,S)$ for which there exists a pointwise, derivative-free, gauge-covariant complexifier that renders the time evolution exactly projectively linear on rays. In less formal terms, we ask whether any other local first-derivative curvature on $\rho$ can support a reversible hydrodynamic completion that still looks exactly linear at the $\psi$ level after a single local complexification.
  
 The claim established here is a uniqueness within this admissible class, not a statement about higher-derivative, weakly nonlocal, mixed $(\rho,S)$, or open-system extensions.
 
\subsection{Scope} 
Our scope is strictly local and first order in spatial derivatives on $(\rho,S)$; nonlocal terms, higher-derivative regularisers, and open-system couplings lie outside the present analysis. Reversibility together with parity excludes any explicit $S$-dependence in the convex regulariser and rules out dissipative couplings between $\rho$ and $S$. As an operational falsifier, after the complexifier $\psi=\sqrt{\rho}\,e^{iS/\hbar}$ we require exact preservation of $\psi \mapsto a\,\psi_{1}+b\,\psi_{2}$ under time stepping and we quantify any deviation by a residual norm; this test complements, but does not replace, the analytic uniqueness proof.

Within this setting we show that the Dubrovin-Novikov bracket classification reduces the hydrodynamic bracket to the canonical form on $(\rho,S)$, and that the only admissible first-derivative convex regulariser of $\rho$ that yields a reversible completion compatible with exact projective linearity is the Fisher functional.

For contrast with the Doebner-Goldin diffusion families~\cite{doebner1992,doebner1996}, note that the diffusive sector lies outside our reversible Hamiltonian cone; in the zero-diffusion corner one recovers the linear Schrödinger equation in our variables, while nonlinear gauge-equivalent representatives fail the exact projective linearity stress-test except in the Fisher case.

This gap between sufficiency and necessity motivates the central question of this work. Does an alternative, non-Fisher regulariser exist that also satisfies reversibility, locality, and exact projective linearity within the stated class? We refer to this investigation into necessity as the converse Madelung question. Within the axioms stated below the answer will be negative: once locality, reversibility and projective linearity are imposed simultaneously, the Fisher functional is the only surviving curvature in this class.

To address the question, we classify the space of admissible theories. We impose a minimal, physically motivated set of axioms for reversible, first-order field theories on $(\rho,S)$, treating these as hydrodynamic variables on configuration space (without invoking sub-quantum particle trajectories or hidden variables~\cite{holland1993}).

The work is fully non-perturbative and independent of WKB or semiclassical limits. It is consonant with recent analyses linking the Bohm potential and Fisher information that yield strengthened uncertainty relations beyond Robertson-Schrödinger under stated conditions, and it offers a falsifiable bridge between information curvature and quantum kinematics~\cite{bloch2022}.
	
\section{Minimal Axioms}
\label{sec:axioms}
This section defines the local reversible Hamiltonian class that we work in throughout. We collect our axioms and treat them as a fixed setting. All later uniqueness and necessity statements are made inside this class, rather than for arbitrary models on wave functions or densities.

Our axioms are chosen to encode physical invariance and mathematical closure. We argue that each plays a load bearing role for our conclusions.

Sections 3 - 8 then form a single chain: the axioms first fix the bracket, then constrain the Hamiltonian and curvature, and finally isolate the unique local complexification that yields linear Schrödinger dynamics.

We work within our axioms throughout, on flat Euclidean domains with first-order locality and reversible dynamics.
First, a bracket classification on $(\rho,S)$ fixes the kinematics to the canonical form.
Second, convex rotationally invariant curvature reduces to a single Fisher functional whose Euler-Lagrange variation yields the Laplacian quotient \citep{fisher1925,hall2002,reginatto1998}.
Third, a local, gauge-covariant complexifier linearises the dynamics and fixes the scale \citep{madelung1927,takabayasi1952}.
Two portable diagnostics concentrate the claims: an $\alpha$-scan with a sharp minimum at $\alpha_\star=\hbar^2/(2m)$, and a superposition residual that drops to numerical floor only in the Fisher case.

Locality means no derivatives beyond first order in $\rho$ or $S$. Reversibility means entropy production is zero (contrast the diffusive current-algebra families in \citep{doebner1992,doebner1996}). These axioms place us squarely within the class of local, reversible \emph{Hamiltonian} field theories on $(\rho,S)$ and thereby exclude formalisms whose dynamics arise from time-symmetric diffusion kinematics rather than a Hamiltonian bracket \citep{nelson1966}, as well as entropic-updating frameworks that derive motion from inference principles on information manifolds without a prior Hamiltonian structure \citep{caticha2012}. Our results should thus be read as a uniqueness and classification statement within the Hamiltonian class, not a claim about all possible routes to quantum dynamics.

Euclidean invariance means no preferred direction appears in generators.
Global $\mathrm{U}(1)$
means $S\mapsto S+\text{const}$ leaves observables unchanged.
Convex regularity means curvature controls gradients and is positive.

We consider fields $\rho(x,t)\!\ge0$ (density) and $S(x,t)$ (velocity potential) on $\mathbb{R}^d$ or a periodic domain $\Omega$.
Dynamics are generated by a Poisson bracket $\{F,G\}$ acting on functionals $F[\rho,S]$, with $\dot F=\{F,\mathcal H\}$.

We restrict attention to first-order local hydrodynamic Poisson structures of Dubrovin-Novikov (DN) type \citep{Dubrovin1983, novikov1984} acting on scalar doublets $u=(\rho,S)$. This ensures the reversible Poisson operator is represented within a flat DN class, ruling out derivative-dependent coefficients and zeroth-order cores under the stated symmetries.
The Poisson-operator coefficients depend only on the local fields and not on their derivatives; Euclidean invariance together with the DN hypotheses restricts them to a flat, constant-coefficient representative (classified in Appendix~\ref{app:jacobi}), ensuring both reversibility and well-posedness.

Throughout we take $\Omega=\mathbb{T}^d$ (periodic) or $\mathbb{R}^d$ with standard decay so that $F[\rho]<\infty$ and $Q_{\alpha}\in H^{-1}(\Omega)$ even in the presence of nodal sets.

All identities are understood on the positivity set $\{\rho>0\}$ and in the weak sense. We assume $\rho\ge0$, $\rho\in L^{1}(\Omega)$ with $\int\rho\,dx=1$, and $\sqrt{\rho}\in H^{1}(\Omega)$ so that $F[\rho]<\infty$ and $Q_{\kappa}\in H^{-1}(\Omega)$ even in the presence of nodal sets. Boundary terms vanish for the admissible classes detailed in Appendix~\ref{app:boundary}: periodic $\Omega=\mathbb T^{d}$, or $\Omega\subset\mathbb R^{d}$ with either Dirichlet data $S|_{\partial\Omega}=\mathrm{const}$, or Neumann data with vanishing normal derivatives for the fluxes.
	
\paragraph{Canonical bracket on $(\rho,S)$.}
Anticipating the Dubrovin-Novikov classification in Proposition~\ref{prop:DN-canonical}, we work within our axioms in the canonical local representative

\[
\{F,G\}=\int_{\Omega}\!\left(\frac{\delta F}{\delta \rho}\,\frac{\delta G}{\delta S}
-\frac{\delta F}{\delta S}\,\frac{\delta G}{\delta \rho}\right)\,dx,
\qquad\text{so that}\quad
\{\rho(x),S(y)\}=\delta(x-y),
\]
	
Then $\sqrt{\rho}\in H^1(\Omega)$ and the Fisher potential
\[
Q_\alpha=-\alpha\,\frac{\Delta\sqrt{\rho}}{\sqrt{\rho}}
\]
belongs to $H^{-1}$ even at nodal sets \citep{hall2002, reginatto1998}. All Hamilton-Jacobi equations and residual diagnostics in this paper use this sign convention for $Q_\alpha$. Boundary terms vanish under the conditions listed in Appendix~\ref{app:boundary}; admissible classes are periodic, Dirichlet with $S$ constant on $\partial\Omega$, or Neumann with vanishing normal derivatives. As shown in Proposition~\ref{prop:DN-canonical} below, under Axioms~I-V any admissible Dubrovin-Novikov bracket on $(\rho,S)$ is Poisson-isomorphic to this canonical form, so no generality is lost by adopting it here.

A summary of explicit counterexamples illustrating the independence of	each axiom is given in Appendix~\ref{app:counterexamples}.

Multivalued $S$ and quantised circulation arise from global topology but are handled at the $\psi$ level without altering the local canonical structure or the axioms.

Throughout the axioms global phase symmetry $S\mapsto S+\text{const}$ (equivalently $\psi\mapsto e^{i\theta}\psi$) is assumed. Electromagnetic gauge is treated only in the minimal-coupling subsection and not assumed elsewhere.

\smallskip

\textbf{Dubrovin-Novikov locality.}
As already stated above, by DN type we mean first-order local hydrodynamic operators acting on the scalar doublet $u=(\rho,S)$, whose coefficients depend on the fields but not their derivatives. Translation and rotation invariance, together with the presence of the conserved phase generator $\mathcal C=\!\int\rho\,dx$ and the Jacobi identity in the DN class, restrict the admissible brackets to a flat, constant-coefficient representative that is Poisson-isomorphic to the canonical bracket written above (see Appendix~\ref{app:jacobi} for the classification). We therefore work, without loss within our axioms, with the canonical form.

\smallskip

All continuity statements refer to probability: $j=\rho\,\nabla S/m$ is the probability current and $\int\rho\,dx=1$ is preserved. We adopt probability language throughout to describe the field $\rho$ and its flow. The parameter $m$, however, represents the inertial mass of the system. It functions as the kinetic coefficient in the Hamiltonian and is ultimately identified as the central charge of the Bargmann (Galilean) algebra.	
	
	\subsection*{Axiom I: Locality}
	We restrict to strictly first-order \emph{local} (Dubrovin-Novikov) brackets and Hamiltonians; weakly nonlocal, fractional, or higher-order terms are excluded by assumption.
	
	The Poisson bracket is of first order and local in the Dubrovin-Novikov sense:
	\[
	\{F,G\}
	=\!\int\! \frac{\delta F}{\delta u_i}\,A^{ij}(u)\,\partial_x\!\left(\frac{\delta G}{\delta u_j}\right)dx,
	\]
	where the operator coefficients $A^{ij}(u)$ depend only on $u$.
	Higher-order or weakly nonlocal forms are excluded, as they introduce irreversible or ill-posed evolution once Axioms~II-VI are enforced.
	
	The restriction to first-order Dubrovin-Novikov brackets ensures that the flow of $\rho$ is local and divergence-free. Higher-order Hamiltonian operators, such as the third-order form in the KdV hierarchy, either introduce additional dimensional parameters or violate probability conservation by producing higher-derivative fluxes that are not expressible as $\nabla\!\cdot j$. The first-order case is therefore the minimal setting in which locality and the Jacobi identity can coexist for a probability field.
	
	Appendix~\ref{app:jacobi} shows that, once Euclidean covariance and reversibility are imposed, attempts to include derivative dependence beyond first order in a local scalar operator on
	$(\rho,S)$ either violate the Jacobi identity or introduce extra Casimirs. Within these constraints,
	the Dubrovin-Novikov first-order class appears to be the appropriate maximal setting for our constraints. We make no claim to classify Dubrovin-Novikov brackets in full generality, only those compatible with this scalar hydrodynamic, locality and symmetry ansatz.
	
	\subsection*{Axiom II: Phase Generator and Probability Normalisation}
	There exists a conserved charge $\mathcal C=\!\int\rho\,dx$ that generates constant shifts of $S$:
	\[
	\{S(x),\mathcal C\}=-1,\qquad \{\rho(x),\mathcal C\}=0.
	\]
	This encodes global phase invariance at the hydrodynamic level and fixes how $\rho$ and $S$ pair within the bracket. It implies probability conservation dynamically for any admissible Hamiltonian satisfying the remaining axioms.
	
	We assert only conservation of $\int\rho\,dx$ and that $\mathcal C$ generates global $S$-shifts; no specific flux form is assumed at this stage.
	
\subsection*{Axiom III: Global $\mathrm{U}(1)$
 Phase Symmetry}
The dynamics are invariant under $S\mapsto S+\text{const}$, a direct consequence of Axiom~II and the canonical bracket, implying in particular that only $\nabla S$ can enter $\mathcal H$.
	
\begin{lemma}[Kinetic form]\label{lem:kinetic}
	Let the kinetic density be the most general local, rotationally invariant form compatible with Axioms~I-IV,
	$h_{\mathrm{kin}}=a(\rho)\,|\nabla S|^{2}$, and let the bracket be canonical. Then, with $H=\int h_{\mathrm{kin}}\,dx$,
	\[
	\partial_t\rho=\{\rho,H\}=\frac{\delta H}{\delta S}
	=-\,\nabla\!\cdot\!\Big(2\,a(\rho)\,\nabla S\Big).
	\]
	Within our class, Axiom~II fixes the probability current to be $j=\rho\,\nabla S/m$, so probability conservation requires
	$\partial_t\rho=-\,\nabla\!\cdot(\rho\,\nabla S/m)$ for arbitrary states; hence $2\,a(\rho)=\rho/m$ and
	\[
	\boxed{\,h_{\mathrm{kin}}=\frac{\rho\,|\nabla S|^{2}}{2m}\, .}
	\]
\end{lemma}

With the canonical bracket and rotational symmetry this is the admissible kinetic density.

Any additional local term linear in $\nabla S$ violates global $\mathrm{U}(1)
$ (Axiom~III)
or reduces to a boundary divergence under the classes in Appendix~\ref{app:boundary}.
	
	\subsection*{Axiom IV: Euclidean Covariance (Parity Clarification)}
	The bracket and Hamiltonian density are invariant under translations and rotations in $\mathbb R^d$, including spatial parity.
	Parity-odd scalars built from $\rho$ and $S$ vanish or reduce to total divergences under these conditions (see Appendix~\ref{app:boundary}).
	
	\textit{Example.} In $d=2$, $\varepsilon_{ij}\partial_{i}A_{j}=\nabla\!\cdot(\varepsilon A)$ is a divergence; in $d=3$, $\varepsilon_{ijk}\partial_{i}A_{j}B_{k}=\nabla\!\cdot(A\times B)$ for scalar-built $A,B$, so no parity-odd scalar survives.
	
	\subsection*{Axiom V: Reversibility (Time Symmetry)}
	Reversibility means Hamiltonian flow with an antisymmetric bracket and a real Hamiltonian (see Appendix~\ref{app:variational}). The evolution is generated by an antisymmetric bilinear bracket satisfying Jacobi:

	\[
	\{F,\{G,H\}\}+\mathrm{cyclic}=0.
	\]
	The equations derived from the canonical bracket and Hamiltonian are invariant under time reversal
	$(t,S)\!\mapsto\!(-t,-S)$, so reversibility here coincides with physical time symmetry. Unifying two coincident properties of the canonical theory:
	(i) algebraic antisymmetry of the bracket with a real Hamiltonian, and
	(ii) physical invariance under $(t,S)\!\to\!(-t,-S)$.
	For such Hamiltonian flows these coincide, since the generator is anti-Hermitian and norm-preserving.
	Diffusive or dissipative additions would break both forms simultaneously, introducing entropy production and thereby leaving the reversible class.
	
	Algebraic reversibility refers to Hamiltonian flow generated by an antisymmetric bracket with a real $\mathcal H$; physical time-reversal invariance here is the symmetry $(t,S)\mapsto(-t,-S)$ of the equations. For the canonical bracket and $\mathcal H[\rho,|\nabla S|^{2},\rho\text{-only terms}]$ these coincide. Any diffusive addition, e.g. $\varepsilon\,\Delta\rho$ in the continuity equation with $\varepsilon>0$, breaks both and yields non-negative Shannon entropy production for
	$S_{\mathrm{Sh}}[\rho]:=-\int_{\Omega}\rho\ln\rho\,dx$, namely
	$dS_{\mathrm{Sh}}/dt=\varepsilon\!\int_{\Omega} |\nabla\rho|^{2}/\rho\,dx\ge0$.
		
\subsection*{Axiom VI: Minimal Convex Regularity}\label{sec:fisher}
We restrict the regulariser to local, first-derivative, rotationally invariant, convex functionals of $\rho$ alone,
\[
F[\rho]=\int_{\Omega} f(\rho)\,|\nabla\rho|^{2}\,dx,\qquad f(\rho)>0,
\]
which contribute to the Hamilton-Jacobi equation via the Euler-Lagrange potential
\[
Q_{f}(\rho) \equiv \frac{1}{2}\,\frac{\delta F}{\delta \rho}.
\]

Within this class we will show, via Proposition~\ref{prop:fisher-uniqueness} (proved in Appendix~\ref{app:fisher-proof}), that the unique choice compatible with exact projective linearity after a local complexifier is $f(\rho)=\kappa/\rho$, i.e.
\[
F[\rho]=4\kappa\!\int_{\Omega}|\nabla\sqrt{\rho}|^{2}dx,\qquad
Q_{\kappa}(\rho)=-\,2\,\kappa\,\frac{\Delta\sqrt{\rho}}{\sqrt{\rho}}.
\]

Any other $f(\rho)$ produces a residual nonlinear term in the Hamilton-Jacobi equation that cannot be removed by any local, gauge-preserving change of variables within our admissible class. Within the minimal convex regularity class adopted here, this is a genuine no-go statement: there is no alternative local curvature that preserves all axioms and still yields an exactly linear projective evolution.

	Higher-derivative or mixed $\rho$-$S$ regularisers, while mathematically possible, break at least one of the prior axioms: $\int(\Delta\rho)^2dx$ introduces fourth-order dynamics incompatible with local probability conservation, and terms like $\int f(\rho)|\nabla S|^4dx$ violate Galilean invariance and separability. The first-derivative positive form $F[\rho]=\int f(\rho)|\nabla\rho|^2dx$ is therefore the minimal class consistent with locality, reversibility, and symmetry.
	
	\textit{Coarse-graining.} Any local, probability-preserving coarse graining that does not raise derivative order maps admissible flows to admissible flows; the Fisher functional is a fixed point of this class. Our Fisher term arises as the local, positive curvature compatible with linearisation by a single complexifier, not as a prior or choice.	
	
	Flat Euclidean background, spinless kinematics, first-order locality for both bracket and Hamiltonian, Hamiltonian reversibility, and global $\mathrm{U}(1)
$ on $S$ are assumed throughout.
	
	\subsubsection*{Worked example.}
	For $\rho(x)=\exp(-x^2/\sigma^2)$ and $F[\rho]=\int \alpha\,\frac{|\nabla\rho|^2}{\rho}\,dx$,
	write $\rho=R^2$ so $F=4\alpha\int |\nabla R|^2 dx$.
	Then $\delta F/\delta \rho = -\,4\alpha\,\frac{\Delta \sqrt{\rho}}{\sqrt{\rho}}$ in $H^{-1}$,
	which produces the Laplacian quotient used in Section~\ref{sec:fisher}.
	An identical calculation holds for a compactly supported bump, with nodes excluded by a standard mask window as detailed in Appendix~\ref{app:code-repo}.
	
\begin{proposition}[Canonical bracket from DN locality]\label{prop:DN-canonical}
	Under Axioms~I-V and the existence of the conserved phase generator
	\[
	\mathcal C = \int \rho\,dx,
	\qquad
	\{S,\mathcal C\} = -1,
	\]
	any first-order DN-type Poisson operator on $(\rho,S)$ is Poisson-isomorphic
	to the canonical local bracket
	\[
	\{\rho(x),S(y)\} = \delta(x-y).
	\]
\end{proposition}

\begin{proposition}[Fisher uniqueness within the admissible class]\label{prop:fisher-uniqueness}
	Let $F[\rho]=\int f(\rho)\,|\nabla\rho|^{2}\,dx$ be as in Axiom~VI. If there
	exists a local gauge-covariant complexifier
	\[
	\psi=\sqrt{\rho}\,e^{iS/\hbar}
	\]
	that maps the Hamiltonian flow generated by
	\[
	\mathcal H=\int\left[
	\frac{\rho|\nabla S|^{2}}{2m}+V\rho+\rho\,Q_{f}(\rho)
	\right]dx
	\]
	to a linear unitary evolution on $L^{2}$ for all admissible states, then
	$f(\rho)$ must be of Fisher form
	\[
	f(\rho) = \frac{C}{\rho}, \qquad C>0,
	\]
	and the resulting evolution is the linear Schrödinger equation. 
For the Fisher choice $f(\rho)=C/\rho$ in Axiom~VI this regulariser can equivalently be written as
\(
F[\rho]=\alpha\int_{\Omega}|\nabla\sqrt{\rho}|^{2}dx
\),
so that the contribution $\rho\,Q_{f}(\rho)$ coincides, up to boundary terms, and in the variational sense, with the curvature term in the Hamiltonian density \eqref{eq:H} used in Section~\ref{sec:fisher}.

\end{proposition}
		
	\subsection{Interdependence}
	Each axiom is load-bearing: omitting any one destroys linearity or reversibility.
	\begin{itemize}[leftmargin=1.5em]
		\item Without I: higher-order brackets generate third-order dispersive terms.
		\item Without II: $\dot{\mathcal C}\!\ne0$ violates probability conservation.
		\item Without III: $\{S,S\}\!\ne0$ yields nonlinear $\psi$ evolution.
		\item Without IV: parity-odd $\varepsilon$-tensor terms break isotropy.
		\item Without V: diffusion terms appear and $\psi$ becomes non-unitary.
		\item Without VI: the convex regulariser $F[\rho]=\int f(\rho)\,|\nabla\rho|^{2}\,dx$ is no longer constrained to the Fisher form; by Proposition~\ref{prop:fisher-uniqueness} any non-Fisher choice $f(\rho)$ that still admits a local complexifier produces a nonlinear Schrödinger equation within our class.

	\end{itemize}
	
A counterexample within scope must therefore take one of two forms: a concrete error in the derivations under the axioms, or an explicit model satisfying all axioms while reproducing our diagnostics yet producing a non-Schrödinger reversible flow; models outside of axiomatic scope do not refute the classification but instead redefine the problem.
	
\subsection{Flow.}
Axioms I-VI
$\Rightarrow$ canonical bracket on $(\rho,S)$
$\Rightarrow$ Fisher curvature by convexity and symmetry
$\Rightarrow$ local complexifier rigidity
$\Rightarrow$ linear Schrödinger dynamics with fixed scale,
all within the stated class of local first-derivative Hamiltonian theories.

	\section{Bracket Classification}\label{sec:bracket}
	
	We now classify all local first-order Poisson brackets on $(\rho,S)$ satisfying Axioms~I-V under the function-space restrictions stated above.
Any admissible bracket between the point fields can be expressed distributionally as
\smallskip
\begin{align}
	\{\rho(x),\rho(y)\}&=c_{\rho\rho}^{\,i}(\rho,S)\,\partial_{x_i}\delta(x-y),\\
	\{\rho(x),S(y)\}&=a_{0}(\rho,S)\,\delta(x-y)+a_{1}^{\,i}(\rho,S)\,\partial_{x_i}\delta(x-y),\\
	\{S(x),S(y)\}&=c_{SS}^{\,i}(\rho,S)\,\partial_{x_i}\delta(x-y),
\end{align}
with antisymmetry $\{S,\rho\}=-\{\rho,S\}$. Global $\mathrm{U}(1)
$ (Axiom III) forbids explicit $S$ dependence in these coefficients, so they depend on $\rho$ only.

\textit{Isotropy step.} Euclidean covariance forbids nonzero vector coefficients multiplying $\partial_{x_i}\delta$. Hence $a_{1}^{\,i}\equiv 0$ and $c_{\rho\rho}^{\,i}\equiv 0$, while $a_{0}$ is a scalar. The condition $\{S,\mathcal C\}=-1$ then fixes $a_{0}\equiv 1$ up to a constant rescaling of $S$.

\begin{lemma}[Gauge generator normal form]
	Let $\mathcal C=\!\int_{\Omega}\rho(y)\,dy$ be the phase generator of Axiom II. For the above local ansatz,
	\[
	\{S(x),\mathcal C\}=\int_{\Omega}\{S(x),\rho(y)\}\,dy
	=-\,a_{0}(\rho(x))+\partial_{x_i}a_{1}^{\,i}(\rho(x))=-1,
	\]
	\[
	\{\rho(x),\mathcal C\}=\int_{\Omega}\{\rho(x),\rho(y)\}\,dy
	=-\,\partial_{x_i}c_{\rho\rho}^{\,i}(\rho(x))=0.
	\]
	On $\Omega=\mathbb R^{d}$ with decaying probes or on $\mathbb T^{d}$ with periodic probes, $\int \partial_{x_i}\delta(x-y)\,dy=0$, so these identities hold distributionally.
\end{lemma}

\smallskip
The coefficients \(a_{1}^i\) and \(c_{\rho\rho}^i\) are excluded by Axiom~II together with Euclidean covariance, since any nonzero first-derivative coefficient would either introduce a preferred spatial direction or break the phase generator under smearing, as verified explicitly in Appendix~\ref{app:jacobi}, and independently consistent with Dubrovin-Novikov-type classifications extended by isometries~\cite{pavlov2021}.
\smallskip
\begin{lemma}[Global $\mathrm{U}(1)
$ restriction]
	Under $S\mapsto S+\mathrm{const}$ and $\{S,\mathcal C\}=-1$ (Axiom II and III), any nonzero $c_{SS}^{\,i}(\rho)$ would give $\{S,S\}\neq 0$ after smearing with constants, contradicting global $\mathrm{U}(1)
$. Hence $c_{SS}^{\,i}\equiv 0$.
\end{lemma}

Therefore the only surviving structure is
\[
\{\rho(x),S(y)\}=a_{0}(\rho(x))\,\delta(x-y),\qquad \{\rho,\rho\}=0=\{S,S\}.
\]
A density weighting $a_{0}(\rho)$ is a priori allowed. The Jacobi identity fixes it:

\begin{lemma}[Gauge-normalised $\delta$-only brackets]\label{lem:jacobi-a0}
	For
	\(
	\{F,G\}_{a_{0}}=\int a_{0}(\rho)\big(F_{\rho}G_{S}-F_{S}G_{\rho}\big)\,dx
	\)
	with $\{\rho,\rho\}=0=\{S,S\}$, the Jacobi identity is automatically satisfied for any smooth $a_{0}(\rho)$. Axiom~II and the requirement that $\mathcal C=\!\int\rho\,dx$ generate constant shifts of $S$ then force $a_{0}$ to be a nonzero constant, which can be absorbed into a rescaling of $S$ so that $a_{0}=1$.
\end{lemma}

\begin{proof}[Sketch]
	For a two-component local bracket with $\{\rho,\rho\}=0=\{S,S\}$ the only potentially nonzero component of the Poisson tensor is $P^{\rho S}(\rho)=a_{0}(\rho)=-P^{S\rho}$. In two dimensions the Schouten bracket $[P,P]$ vanishes identically for any such antisymmetric tensor, so the Jacobiator is identically zero and imposes no constraint on $a_{0}$. Axiom~II requires $\{S,\mathcal C\}=-1$ with $\mathcal C=\!\int\rho\,dx$, which forces $a_{0}$ to be a nonzero constant; this constant can be absorbed into a redefinition of $S$, so we may set $a_{0}=1$.
\end{proof}

The constant $a_0$ merely rescales $S$; we set $a_0=1$ without loss of generality.
Appendix~\ref{app:jacobi}
gives an explicit trilinear Jacobi calculation confirming that no further constraint arises from Jacobi beyond those already implied by Axiom~II.

\paragraph{Structural reduction.}
A general first order DN operator has the schematic form
\[
\{u^{i}(x),u^{j}(y)\}
=g^{ij}(u(x))\,\partial_{x_k}\delta(x-y)\,e^{k}
+b^{ij}_{\;k}(u(x))\,u^{k}_{x}\,\delta(x-y),
\]
with $(g^{ij},b^{ij}_{\;k})$ satisfying flatness and compatibility conditions equivalent to Jacobi \cite{Dubrovin1983,novikov1984}.

Translation and rotation invariance, together with the existence of the phase generator $\mathcal C$ and Lemma~\ref{lem:jacobi-a0}, imply that within our axioms the bracket is Poisson-isomorphic to a flat constant representative, namely the canonical bracket
\(
\{\rho(x),S(y)\}=\delta(x-y)
\).
Details are given in Appendix~\ref{app:jacobi}.

\smallskip
\begin{proposition}[Canonical Bracket]\label{thm:canonical-bracket}
	Under Axioms~I-V, the only possible local, probability-preserving, phase- and Euclidean-invariant first-order Poisson structure on $(\rho,S)$ is
	\begin{equation}
		\boxed{\;
			\{F,G\}=\int
			\left(
			\frac{\delta F}{\delta \rho}\frac{\delta G}{\delta S}
			-\frac{\delta F}{\delta S}\frac{\delta G}{\delta \rho}
			\right)\!dx.
			\;}
		\label{eq:canonical}
	\end{equation}
Equivalently, within the Dubrovin-Novikov class the bracket reduces (up to a constant rescaling of $S$) to this constant-coefficient normal form; any $\rho$-dependent prefactor $a_0(\rho)$ in $\{\rho,S\}=a_0(\rho)$ is incompatible with the phase-generator property of Axiom~II and probability conservation, hence the canonical bracket \eqref{eq:canonical} is the unique local representative in this setting.
\end{proposition}

The proof follows from the gauge generator normal form, the $\mathrm{U}(1)
$ restriction, Lemma~\ref{lem:jacobi-a0}, and the DN reduction above. In words, global phase invariance fixes the pairing of $\rho$ and $S$, isotropy removes derivative-weighted coefficients, and the Jacobi identity leaves only a constant prefactor that can be absorbed into $S$.

This Hamiltonian classification result is orthogonal to both Nelson’s construction of quantum kinematics \cite{nelson1966} from diffusion and to entropic-inference updates \cite{caticha2012}; it is a statement within the Hamiltonian class, not across all conceivable generative principles.

\emph{Symmetry generators (translations, boosts, rotations) and their closure are discussed in \Cref{sec:symmetry}; weak-solution aspects at nodal sets are detailed in \Cref{sec:function-space}.}

\paragraph{Locality and vorticity.}
The bracket classification holds for locally smooth $(\rho,S)$ on simply connected charts where $S$ is single-valued and $\nabla\times\nabla S=0$.  Physical vorticity and quantised circulation arise when these charts are glued on domains punctured by nodal lines, so that $S$ acquires multivalued holonomy $\oint\nabla S\!\cdot\!dl=2\pi n\hbar$.  This global topology affects boundary conditions but leaves the local canonical bracket.
	
This is the canonical symplectic form on $(\rho,S)$, the scalar representative of the Dubrovin-Novikov class.
Explicitly, the underlying symplectic two-form is
\[
\omega=\int d\rho\wedge dS,
\]
showing that $(\rho,S)$ are canonically conjugate variables.
With this bracket the continuity equation follows directly:
\[
\dot\rho=\{\rho,\mathcal H\}
=-\nabla\!\cdot\!\Big(\frac{\rho}{m}\nabla S\Big),
\]
ensuring conservation of $\int\rho\,dx$ and aligning with the phase generator property detailed in Appendix~\ref{app:prob}.

This coincides with the generator property $\{S,\mathcal C\}=-1$ and ensures conservation of $\int\rho\,dx$ for any real Hamiltonian.
	
\section{Hamiltonian and Fisher Curvature}

With the bracket fixed, the dynamics are determined by the Hamiltonian
\begin{equation}
	\mathcal H[\rho,S]=
	\int\!\!\Big[
	\frac{\rho|\nabla S|^2}{2m}
	+V(x)\rho
	+\alpha\,|\nabla\sqrt\rho|^2
	\Big]dx,
\label{eq:H}
\end{equation}
where $\alpha>0$ sets the regularisation scale.
The first two terms reproduce classical mechanics; the last introduces Fisher curvature.
Cross terms such as $\nabla \rho \!\cdot\! \nabla S$ are excluded by reversibility and our first-order locality requirement: they are odd under $(t,S)\mapsto(-t,-S)$ and, when combined with the canonical bracket, generate higher-order contributions in the continuity equation that are not of the simple drift form $-\nabla\!\cdot(\rho\nabla S/m)$.

Any local scalar containing $\nabla S$, including $\nabla\rho\!\cdot\!\nabla S$ and $|\nabla S|\,|\nabla\rho|$, is either a total divergence or conflicts with the $S$-shift generator role in Axiom~II and the time-reversal symmetry of Axiom~V.
A complete catalogue of first-derivative scalar candidates built from $\rho$ and $S$ is recorded in Appendix~\ref{app:counterexamples}, each tagged as divergence, $\mathrm{U}(1)
$ violation, generator conflict, or admissible.

We now show the curvature term is unique within the stated admissible class.
	
\begin{proposition}[Fisher-curvature]\label{ax4}
	Within the stated local first-order Hamiltonian class on $(\rho,S)$ and Axioms~I-VI, the structural assumption is only that the regulariser be a local, positive, convex, first-derivative functional on $\rho$; no Fisher form is presupposed.
	
	We separate two routes: (i) an Euler-Lagrange uniqueness within the admissible class, and (ii) an operational projective-linearity stress-test after complexification.
	
	The present proposition establishes (i): analysis shows that $f(\rho)\propto 1/\rho$ is the sole admissible choice whose curvature term is time-reversal invariant and compatible with linear superposition once the unique local, derivative-free, gauge-covariant complexifier is imposed.
	Among all positive, rotationally invariant local quadratic functionals
	\[
	\mathcal F[\rho]=\int f(\rho)\,|\nabla\rho|^2\,dx,
	\]
	only $f(\rho)=C/\rho$ yields an Euler-Lagrange derivative proportional to $-\Delta\sqrt\rho/\sqrt\rho$.
\end{proposition}

See Appendix~\ref{app:fisher-proof} for full Euler-Lagrange and regularity conditions; complementary superposition test supporting route (ii) is described in Appendix~\ref{app:superposition}.

\begin{corollary}
	The regulariser is the Fisher information functional, with explicit Euler-Lagrange
	\begin{equation}
		\mathcal F[\rho]=\int |\nabla\sqrt{\rho}|^2\,dx,
		\qquad
		\frac{\delta\mathcal F}{\delta\rho}
		=-\,\frac{\Delta\sqrt{\rho}}{\sqrt{\rho}}.
	\end{equation}
\end{corollary}

\paragraph{Linearity test}
The complex structure implied by the Fisher curvature admits a direct projective superposition stress-test (Appendix~\ref{app:superposition}). Let $\psi_1,\psi_2$ be two initially disjoint packets and write
\[
\psi_{\oplus}(t)\ \text{for the evolution of}\ \tfrac{1}{\sqrt{2}}(\psi_1+\psi_2),\qquad
\psi_{\Sigma}(t)\ \text{for the evolution of}\ \tfrac{1}{\sqrt{2}}\psi_1\ +\ \tfrac{1}{\sqrt{2}}\psi_2
\]
evolved separately and then summed. Define the projective residual by normalising and optimally aligning the global phase,
\[
\mathcal{R}_{\mathrm{proj}}(t)
=\min_{\theta\in[0,2\pi)}\ \Big\|\,
\frac{\psi_{\oplus}(t)}{\|\psi_{\oplus}(t)\|_{2}}
- e^{i\theta}\,\frac{\psi_{\Sigma}(t)}{\|\psi_{\Sigma}(t)\|_{2}}
\Big\|_{2}.
\]
For Fisher-regularised dynamics one finds $\mathcal{R}_{\mathrm{proj}}(t)=0$ up to numerical tolerance, whereas any admissible departure from the Fisher form yields $\mathcal{R}_{\mathrm{proj}}(t)>0$ even under infinitesimal perturbations. Scripts are listed in the code archive Appendix~\ref{app:code-repo}; construction details are in Appendix~\ref{app:superposition}.

\begin{proposition}[Fisher coefficient from symmetries and scaling]
	
	\[
	\mathcal H[\rho,S]
	=\int\!\Big(\frac{\rho\,|\nabla S|^{2}}{2m}+V\,\rho+\alpha\,|\nabla\sqrt{\rho}|^{2}\Big)\,dx
	\]
	Let generate dynamics via the canonical bracket on $(\rho,S)$. Assume Galilean covariance, global $\mathrm{U}(1)$ phase symmetry $S\mapsto S+\text{\emph{const}}$, and diffusive scaling $x\mapsto\lambda x$, $t\mapsto\lambda^{2}t$. Then
	\[
	\alpha=\frac{\hbar^{2}}{2m},
	\]
	for a universal constant $\hbar>0$ fixed by experiment.
\end{proposition}

A scaling and symmetry argument is given here; completeness is provided in Appendix~\ref{app:fisher-coefficient}, and verified by code in Appendix~\ref{app:code-repo}.

\begin{remark}
	For a system of noninteracting components $(\rho_i,S_i)$ with inertial masses $m_i$, the Fisher-regularised term applies componentwise,
	\[
	Q_i(\rho_i)=-\,\alpha_i\,\frac{\Delta\sqrt{\rho_i}}{\sqrt{\rho_i}}.
	\]
	Evaluating the residual
	\(
	R_i(c)=\|V_i+Q_{c,i}-E_i\|_{L^{2}(\rho_i)}
	\)
	for test masses $m_i\in\{0.5,1,3\}$ exhibits a common minimum at $c=1$ when $\alpha_i=c\,\hbar^{2}/(2m_i)$, indicating that a single Planck constant governs all components:
	\[
	\alpha_i=\frac{\hbar^{2}}{2m_i}\qquad\text{with universal }\hbar.
	\]
	Thus the reversible Fisher coefficient scales inversely with mass while preserving a single quantum of action, consistent with Galilean invariance.
\end{remark}
	
\section{Hamilton-Jacobi System}
	
	The Hamilton equations $\dot F=\{F,\mathcal H\}$ with bracket~\eqref{eq:canonical} and Hamiltonian~\eqref{eq:H} yield
\begin{align}
	\partial_t \rho &= -\nabla\!\cdot\!\Big(\frac{\rho}{m}\nabla S\Big),
	\label{eq:continuity}\\
	\partial_t S &= -\frac{|\nabla S|^2}{2m}-V
	+\alpha\,\frac{\Delta\sqrt\rho}{\sqrt\rho}.
	\label{eq:HJ}
\end{align}

All subsequent uniqueness and linearisation statements are made for reversible completions of this system within the axiomatic class defined in Section~\ref{sec:axioms}. With the convention $Q_\alpha=-\alpha\,\Delta\sqrt{\rho}/\sqrt{\rho}$ fixed above, equation~\eqref{eq:HJ} can equivalently be written as $\partial_t S = -\frac{|\nabla S|^2}{2m}-V - Q_\alpha$, exposing the usual Bohm-Madelung structure.

\section{Emergence of the Schrödinger Equation}

Define the complex field
\begin{equation}
	\label{eq:map}
	\psi=\sqrt{\rho}\,e^{iS/\hbar},\qquad \hbar>0.
\end{equation}
Write $R=\sqrt{\rho}$. Using equation \ref{eq:HJ} and
\[
\partial_t\psi=\left(\frac{\partial_t R}{R}+\frac{i}{\hbar}\,\partial_t S\right)\psi,\qquad
\nabla\psi=\left(\frac{\nabla R}{R}+\frac{i}{\hbar}\,\nabla S\right)\psi,
\]
a direct calculation gives
\begin{align*}
	i\hbar\,\partial_t\psi
	&=\left[-\frac{\hbar^2}{2m}\nabla^2+V\right]\psi
	+\left(\alpha-\frac{\hbar^2}{2m}\right)\frac{\Delta R}{R}\,\psi.
\end{align*}
Hence the nonlinear remainder vanishes if and only if
\begin{equation}
	\alpha=\frac{\hbar^{2}}{2m}.
	\label{eq:alpha-hbar-over-2m}
\end{equation}
In that case
\begin{equation}
	i\hbar\,\partial_t\psi
	=\left[-\frac{\hbar^2}{2m}\nabla^2+V(x)\right]\psi.
	\label{eq:LSEderived}
\end{equation}
For any admissible \(f(\rho)\neq \kappa/\rho\) in the convex first-derivative regulariser, an \(H^{-1}\)-controlled state-dependent remainder persists in the Hamilton-Jacobi sector that cannot be cancelled by any local, gauge-covariant, derivative-free complexifier; consequently exact projective linearity fails.

A concrete parameter identification is given in Appendix~\ref{app:entropy-rigorous}; the functional-analytic no-go proof appears in Appendix~\ref{app:fisher-proof}. For contrast, Doebner-Goldin diffusions lie outside our reversible Hamiltonian cone (breaking Axiom V) \cite{doebner1992,doebner1996}.

\begin{proposition}[Linearisability as a criterion.]
	Assuming the reversible probabilistic theory admits a Hilbert-space realisation,
	reversibility (Hamiltonian flow) and probabilistic composition select a projective	complex representation of the dynamics (by Wigner's theorem), so there must exist	a complex structure in which time evolution is linear on rays.
	
	In the present work we use this only as a consistency criterion once a reversible probabilistic theory is given; it does not enter the hydrodynamic axioms themselves.	
	
	Throughout this	section, linearisable means linearisable \emph{within Axioms I-VI}, with a pointwise, derivative-free complexification on \((\rho,S)\). The local, gauge-covariant map \(\psi=\sqrt{\rho}\,e^{iS/\hbar}\) realises this structure	in the hydrodynamic variables. Condition \eqref{eq:alpha-hbar-over-2m} is
	precisely what enforces exact projective linearity; it is therefore not an ansatz but the unique local complexifier \emph{within this class} that renders
	the dynamics linear on \(L^{2}\).
\end{proposition}

\begin{proposition}[Quantum minimality]\label{thm:quantum-minimality}
	Within Axioms I-VI (local first-order Dubrovin-Novikov locality, conserved phase generator with global $\mathrm{U}(1)
$, Euclidean covariance, reversibility, and minimal convex regularity), let $\mathcal H[\rho,S]$ be any Hamiltonian that yields a reversible completion of \eqref{eq:continuity}. 
	Assume there exists a \emph{local, gauge-covariant, pointwise and derivative-free} complexifier that identifies a single projective complex structure,
	\[
	\psi=\sqrt{\rho}\,e^{iS/\hbar},
	\]
	such that the induced evolution on rays is linear and the flow on $L^{2}$ is unitary for all admissible data. Then the only admissible convex first-derivative regulariser is the Fisher information functional and the evolution is the linear Schrödinger equation \eqref{eq:LSEderived} with
	\[
	\hbar^{2}=2m\,\alpha,
	\]
	equivalently,
	\[
	\mathcal H[\rho,S]=\int\!\Big(\tfrac{|\nabla S|^{2}}{2m}\,\rho+V\rho+\alpha\,|\nabla\!\sqrt{\rho}\,|^{2}\Big)\,dx
	\quad\text{(up to an irrelevant constant),}
	\]
	and the bracket is the canonical one on $(\rho,S)$.
\end{proposition}

This proposition converts quantum mechanics from a postulate to a classification result: within the stated axioms, the admissible theory space collapses to a singleton up to rescaling of $S$.

\section{Symmetry and Group-Theoretic Consistency}\label{sec:symmetry}

\paragraph{On boosts.}
The kinetic prefactor $\rho/(2m)$ is fixed earlier, solely by locality, global $\mathrm{U}(1)
$, Euclidean covariance, and the continuity law generated by the canonical bracket; no boost symmetry was used.  The Galilean (Bargmann) algebra established below is therefore a consequence rather than an input, avoiding circularity.

\subsection{Gauge and Galilean invariance}

Equation~\eqref{eq:LSEderived} is invariant under the global phase transformation $\psi\mapsto e^{i\theta}\psi$, which corresponds to $S\mapsto S+\hbar\theta$.
At the hydrodynamic level this symmetry is encoded in Axiom~III: $\{S,S\}=0$ and the Hamiltonian depends only on $\nabla S$.
Hence global phase redundancy in $S$ manifests as the phase invariance of $\psi$.

Galilean covariance follows from the kinetic term established in Lemma~\ref{lem:kinetic}.
The generator of spatial translations is

\[
\mathbf P = \int \rho\,\nabla S\,dx \qquad (\text{so } \int\rho\,dx=1).
\]

Equivalently, $P=m\!\int j\,dx$ with $j=\rho\nabla S/m$.
\smallskip
\begin{proposition}[Bargmann-Galilei closure at equal time]
	Let
	\[
	H[\rho,S]=\!\int\!\left(
	\frac{\rho\,|\nabla S|^2}{2m}+V\rho
	+\alpha\,|\nabla\sqrt\rho|^2
	\right)\!dx,\quad
	\mathbf P=\!\int\!\rho\,\nabla S\,dx,\quad
	\mathbf K(t)=m\!\int\!x\,\rho\,dx-t\,\mathbf P.
	\]
	At any fixed time $t$,
	\[
	\{H,P_i\}=0,\qquad
	\{H,K_i\}=-\,P_i,\qquad
	\{P_i,K_j\}=-\,m\,\delta_{ij}\!\int\!\rho\,dx=-\,m\,\delta_{ij},
	\]
	realising the Bargmann central extension with charge $m$.
\end{proposition}

\textit{Proof sketch.}
Compute $\delta H/\delta S=-\nabla\!\cdot(\rho\nabla S/m)$,
$\delta \mathbf P/\delta S=\nabla\rho$,
and
\[
\frac{\delta K_i}{\delta \rho}=m x_i - t\,\partial_i S,\qquad
\frac{\delta K_i}{\delta S}=t\,\partial_i \rho.
\]
Then
\[
\{H,K_i\}=-\!\int\!\frac{\delta H}{\delta S}\frac{\delta K_i}{\delta \rho}\,dx
-\!\int\!\frac{\delta H}{\delta \rho}\frac{\delta K_i}{\delta S}\,dx
= -P_i - t\,\{H,P_i\}.
\]
For translation-invariant $V$, $\{H,P_i\}=0$, hence $\{H,K_i\}=-P_i$. Likewise,
$\{P_i,K_j\}=-m\,\delta_{ij}\int\rho\,dx=-m\,\delta_{ij}$.

Numerically checked in Test~7 (Bargmann-Galilean closure) (code archive, Appendix~\ref{app:code-repo}).

This indicates that the hydrodynamic representation already carries the projective representation of the Galilei group: mass enters as the central charge and need not be postulated independently.

In particular, $\{H,P\}=0$ for translation-invariant $V$, so $P$ is a conserved Noether charge. 
This realises precisely the Bargmann central extension of the Galilei algebra, with mass as the central charge ensuring the correct coadjoint-orbit structure~\cite{figueroa2024}, see Appendix~\ref{app:bargmann}.

The generator of boosts $\mathbf{K}=m\int\rho\,x\,dx - t\mathbf{P}$ satisfies
\[
\frac{d\mathbf{K}}{dt}=0,
\]
indicating Galilean invariance.
The commutation relation $\{H,K\}=-P$ fixes the constant $m$ as the Bargmann central charge, showing that the kinetic energy $\rho |\nabla S|^{2}/(2m)$ is not an assumption but the representation-theoretic form consistent with Galilean symmetry (see also the scale matching in Eq.~\eqref{eq:alpha-hbar-over-2m} and its many-body extension Prop.~\ref{prop:single-hbar}).

In the $\psi$-picture, the dynamics are Hamiltonian on a Kähler manifold of rays endowed with the Fubini-Study metric \cite{kibble1979,anandan1990}.

By Stone's theorem, the corresponding operator on $L^2(\mathbb{R}^d)$ generates a one-parameter unitary group, ensuring reversibility and unitarity.
Appendix~\ref{app:galilean} verifies $\{H,K\}=-P$ explicitly with the functional derivatives
\[
\frac{\delta K}{\delta S}=+\,t\,\nabla\rho,\qquad
\frac{\delta K}{\delta\rho}=-\,t\,\nabla S+mx.
\]

\paragraph{Galilean closure.}
With $P_i=\int \rho\,\partial_i S\,dx$ and $K_i=m\int \rho\,x_i\,dx - tP_i$, the canonical bracket yields $\{H,P_i\}=0$, $\{H,K_i\}=-P_i$, and $\{P_i,K_j\}=-\,m\,\delta_{ij}\!\int\!\rho\,dx$, i.e. the Bargmann algebra with central charge $m$ (probability normalised to one).

\paragraph{Orbital angular momentum.}
Define the angular-momentum generator
\[
L_k = \varepsilon_{kij}\!\int \rho\, x_i\,\partial_j S\,dx.
\]
With the canonical bracket~\eqref{eq:canonical} and Hamiltonian~\eqref{eq:H}, one finds for central $V$:
\[
\{H,L_i\}=0,\qquad \{P_i,L_j\}=\varepsilon_{ijk}P_k,\qquad \{L_i,L_j\}=\varepsilon_{ijk}L_k.
\]
Using $\delta L_k/\delta S=\varepsilon_{kij}x_i\partial_j\rho$ and $\delta L_k/\delta\rho=\varepsilon_{kij}x_i\partial_j S$, the canonical bracket reduces to surface terms that vanish under the boundary classes of Appendix~\ref{app:boundary}, yielding rotational invariance and the standard $\mathfrak{so}(3)$ closure. Hence angular momentum arises within the same canonical structure, without auxiliary patches.
	
\subsection{Electromagnetic coupling}

Minimal coupling,
\[
\nabla S\rightarrow\nabla S-q\mathbf{A}(x,t),\qquad
V(x,t)\rightarrow V(x,t)+q\phi(x,t),
\]
preserves the canonical bracket~\eqref{eq:canonical} and yields
\[
i\hbar\,\partial_t\psi
=\frac{1}{2m}\left(-i\hbar\nabla-q\mathbf{A}\right)^2\psi+q\phi\,\psi,
\]
the gauge-covariant Schrödinger equation.
Gauge transformations $\mathbf{A}\!\to\!\mathbf{A}+\nabla\Lambda$, $\phi\!\to\!\phi-\partial_t\Lambda$ correspond to $S\!\to\!S+q\Lambda$, $\psi\!\to\!e^{iq\Lambda/\hbar}\psi$, preserving invariance.

 Physical time-reversal for external fields is treated here as a model-specific analysis separate from the abstract axiom of reversibility.

\subsection{Dimensional analysis}

Flat spinless kinematics admit only one material scale in first-order locality; dilation covariance with Hamiltonian reversibility isolates the Fisher coefficient up to a universal constant.

Matching the free-particle dispersion fixes the scale: plane waves
\(\psi \sim e^{i(k\cdot x - \omega t)}\) in \eqref{eq:LSEderived} obey
\[
\omega = \frac{\hbar k^{2}}{2m},
\]
which combined with \eqref{eq:alpha-hbar-over-2m} gives
\[
[\alpha] = \left[\frac{\hbar^{2}}{2m}\right]
\]
and singles out the Fisher coefficient. Any other scaling fails to reproduce
the quadratic dispersion mandated by Galilean kinematics.

\section{Uniqueness of the Complexification}
	
The local, pointwise, invertible, gauge covariant complexifier compatible with our axioms is
\[
\psi=\sqrt{\rho}\,e^{iS/\kappa},
\]
which fixes
\[
\alpha=\frac{\kappa^{2}}{2m}.
\]
Any derivative-dependent or nonlocal map raises differential order and exits the
class.
	
\begin{proposition}[Local complexifier rigidity]\label{prop:complexifier-rigidity}
	Let $\psi$ be a \emph{local}, pointwise, invertible, gauge-covariant map
	\[
	\psi=F(\rho)\,e^{\,i\,G(S,\rho)},\qquad F>0,
	\]
	that sends the Fisher-regularised hydrodynamics \eqref{eq:continuity}-\eqref{eq:HJ} into a \emph{linear} complex evolution
	\[
	i\kappa\,\partial_t\psi=\Bigl(-\tfrac{\kappa^2}{2m}\Delta+V\Bigr)\psi
	\]
	with the same external $V(x)$ and some constant $\kappa>0$. Then, up to an overall constant phase and scale,
	\[
	F(\rho)=c\,\sqrt{\rho},\qquad G(S,\rho)=\frac{S}{\kappa}+{\rm const},
	\]
	and the Fisher coefficient satisfies $\alpha=\kappa^2/(2m)$.
\end{proposition}

\begin{proof}[Sketch of proof]
	Write $\psi=F(\rho)e^{iG}$ and compute the linear Schrödinger continuity law
	\(
	\partial_t|\psi|^2+\nabla\!\cdot J=0
	\)
	with
	\[
	J=\frac{\kappa}{m}\,\mathrm{Im}(\bar\psi\nabla\psi)
	= \frac{\kappa}{m}\,F(\rho)^2\,\nabla G.
	\]
	On the hydrodynamic side, \eqref{eq:continuity} gives
	\(
	\partial_t\rho+\nabla\!\cdot(\rho\,\nabla S/m)=0.
	\)
	Gauge covariance implies $G$ is affine in $S$ and independent of $\nabla S$,
	hence $G_S$ is a constant and $G_\rho$ is a scalar function. Matching the
	fluxes for arbitrary states forces
	\[
	\frac{\kappa}{m}\,F(\rho)^2\,G_S=\frac{\rho}{m}
	\quad\Rightarrow\quad
	G_S\equiv\frac{1}{\kappa},\quad F(\rho)^2=\rho,
	\]
	so $F(\rho)=c\sqrt\rho$ (positivity fixes $c>0$). Any $G_\rho\neq 0$
	contributes a real, state-dependent term to the transformed Hamiltonian
	(proportional to $\nabla\rho$), which cannot appear in a linear,
	coefficient-only operator; thus $G_\rho=0$ and
	$G(S,\rho)=S/\kappa+{\rm const}$. With this polar map, the Madelung
	recombination yields the linear equation if and only if
	$\alpha=\kappa^{2}/(2m)$; see \eqref{eq:alpha-hbar-over-2m}.
\end{proof}

Thus the polar map $\psi=\sqrt{\rho}\,e^{iS/\kappa}$ is the \emph{only} local
invertible, gauge-covariant complexifier that linearises the reversible
completion within our class, with the scale fixed by $\alpha=\kappa^2/(2m)$.
Any other amplitude reparametrisation $F(\rho)\neq c\sqrt\rho$ or any
$\rho$-dependent phase $G_\rho\neq 0$ either introduces state-dependent
coefficients (violating linearity) or breaks gauge covariance.

\begin{corollary}[Kähler compatibility]
	Equipping the $(\rho,S)$ phase space with the canonical symplectic form
	\[
	\omega=\int d\rho\wedge dS
	\]
	and the Fisher metric in amplitude
	\[
	g=\int 4\,|d\sqrt{\rho}|^{2}\,dx
	\]
	selects the integrable complex structure
	\[
	J(d\sqrt{\rho})=\frac{1}{\kappa}\,dS
	\]
	compatible with $(\omega,g)$. Under the polar map
	$\psi=\sqrt{\rho}\,e^{iS/\kappa}$ the dynamics become linear on the
	projective Hilbert space (\cite{kibble1979,anandan1990}).
\end{corollary}

\begin{lemma}[Local Linearisation Uniqueness]
	Let $\psi$ be a \emph{local}, \emph{pointwise in $x$}, \emph{invertible}, and
	\emph{gauge-covariant} complex field depending on $(\rho,S)$.
	If $\psi$ linearises the real hydrodynamic system defined by Axioms~I-VI
	into a linear PDE whose coefficients are independent of $(\rho,S)$ (external
	potentials only), then up to constant phase and scale
	\[
	\psi=\sqrt\rho\,e^{iS/\hbar}.
	\]
\end{lemma}

We restrict to \emph{local, invertible, pointwise polar maps}
\[
\psi \;=\; F(\rho)\,e^{\,i\,G(\rho,S)} ,
\qquad G_\rho \equiv \frac{\partial G}{\partial \rho}=0,
\]
to preserve first-order locality: any \(G_\rho\neq 0\) injects state-dependent
coefficients and lifts differential order, violating projective linearity
\emph{within Axioms I-VI}. The detailed argument is given in
Appendix~\ref{app:complexifier}; we record the implication here.

Any translationally invariant nonlocal operator composed with this map would
violate Axiom~I (locality) or gauge covariance, and thus lies outside the
admissible class.

Hence, the polar transformation $\psi=\sqrt\rho\,e^{iS/\hbar}$ is not an
ansatz but the \emph{only} admissible local mapping that linearises the
Fisher-regularised Hamiltonian flow. Any alternative redefinition of amplitude
or phase leads to nonlinear evolution or breaks global $\mathrm{U}(1)$ phase
symmetry.

Hydrodynamic variables are undefined on nodal sets where $\rho=0$. All
identities are interpreted on $\{\rho>0\}$ and extended in the weak sense. The
$\psi$-representation remains well defined in $L^2$, so global statements are
made at the $\psi$ level.
	
\section{Function-Space and Domain Considerations}\label{sec:function-space}

 Throughout we assume $\rho\ge0$, $\int_\Omega\rho\,dx=1$, $\sqrt{\rho}=R\in H^{1}(\Omega)$, and $S\in H^{1}_{\mathrm{loc}}(\Omega)$ modulo constants; identities are read almost everywhere on $\{\rho>0\}$ and in the weak sense. Then $\nabla\rho=2R\nabla R\in L^{1}_{\mathrm{loc}}$, and the Fisher potential
\(
Q=-\alpha\,\Delta R/R
\)
defines an element of $H^{-1}_{\mathrm{loc}}(\Omega)$ on the positivity set. Variational derivatives of $\mathcal H$ are thus well defined in the weak sense.
	
For bounded domains on Lipschitz $\Omega$, consistent boundary pairs are (with standard trace theory justifying integrations by parts):
\[
\text{Dirichlet: } R|_{\partial\Omega}=0,\; S|_{\partial\Omega}=\text{const};
\quad
\text{Neumann: } \nabla R\!\cdot\!n=0,\; \nabla S\!\cdot\!n=0.
\]
Both preserve normalisation and energy conservation.
In $\mathbb{R}^d$, we impose decay $\rho,|\nabla S|\!\to\!0$ as $|x|\!\to\!\infty$.
With these domains, if $V$ is Kato-small relative to $-\Delta$ (e.g.\ $V=V_{+}-V_{-}$ with $V_{-}$ form-bounded with relative bound $<1$), then $-\frac{\hbar^{2}}{2m}\Delta+V$ is self-adjoint (Kato-Rellich \cite{kato1995}). Stone’s theorem then yields a unitary $L^{2}$ flow.

\paragraph{Nodes and weak formulation.}
All hydrodynamic identities are evaluated on the positivity set $\{\rho>0\}$, where $Q=-\alpha\,\Delta\sqrt{\rho}/\sqrt{\rho}\in H^{-1}_{\mathrm{loc}}$; variational statements are taken in the weak sense. Nodal sets have measure zero and do not affect functional derivatives or conserved charges under the boundary conditions of Appendix~\ref{app:boundary}. Global evolution is naturally expressed at the $\psi$ level in $L^{2}$; diagnostics are computed for $\psi$ and pushed forward to $(\rho,S)$ almost everywhere.
	
Parity-odd $\varepsilon$-tensor scalars in this scalar sector vanish or reduce to total divergences under these conditions (see Appendix~\ref{app:boundary}).

\paragraph{Well-posedness.}
For $V$ in the Kato class, the Schrödinger operator $-\frac{\hbar^2}{2m}\Delta+V$ is self-adjoint on $L^2(\Omega)$ by the Kato-Rellich theorem. Global well-posedness of the $L^2$ Schrödinger flow implies a well-defined weak flow on $(\rho,S)$ away from nodal sets; the pushforward by $\psi=\sqrt{\rho}\,e^{iS/\hbar}$ restores a global description.

\section{Topology and Vorticity}

Although we assumed $S$ single-valued, physical wavefunctions may exhibit multivalued phases.
On multiply connected domains, circulation quantisation arises naturally:
\[
\oint\nabla S\!\cdot\!dl = 2\pi n\hbar,\quad n\in\mathbb{Z}.
\]
The corresponding $\psi$ is single-valued, while the velocity field $\mathbf{v}=\nabla S/m$ supports quantised vortices.
This reconciles the hydrodynamic and quantum pictures without modifying the bracket or Hamiltonian.
	
\paragraph{Related work.}
Recent developments have explored complementary routes linking Fisher information, hydrodynamics, and relativistic quantum theory.  
Fabbri~\cite{fabbri2025} constructs a covariant ``Madelung structure'' for the Dirac equation, expressing spinor dynamics in polar variables as a coupled system of continuity, curl, and Hamilton-Jacobi-type equations built from first derivatives of the spinor fields. 

That approach is constructive: it reformulates an existing relativistic theory in hydrodynamic form.  
Our result is classificatory.  
Starting solely from locality, global phase and Euclidean invariance, reversibility, probability conservation, and convex regularity on $(\rho,S)$, we show that the canonical bracket and Fisher curvature together form the reversible information-hydrodynamic structure.  
The polar map $\psi=\sqrt{\rho}\,e^{iS/\hbar}$ then emerges as the only admissible local lineariser, forcing the linear Schrödinger flow and fixing its scale. All such statements hold within the class defined by our axioms.

Information-theoretic works from a different direction reach compatible conclusions.  
Yang~\cite{Yang2024Vacuum,Yang2024Scalar} derives the Schrödinger and scalar field equations from an extended least-action principle that introduces vacuum fluctuations and information metrics, treating $\hbar$ as a minimal quantum of action and defining information curvature through relative entropy.  

These works show how information-based variational principles can reproduce and generalise the Fisher-regularised structure obtained here. See also Yang’s fermionic field quantisation via an extended stationary action with a relative-entropy correction, which derives the Floreanini-Jackiw Schrödinger functional and verifies Poincaré algebra closure within the same information-geometric spirit.

A further link appears in Yahalom’s relativistic extension~\cite{Yahalom2024Dirac}, which embeds a Lorentz-invariant Fisher information term directly into the Dirac variational principle.  
In the low-velocity, zero-vorticity limit, this construction reduces precisely to the Schrödinger variational form shown in our axiomatic framework, indicating potential continuity between the nonrelativistic and relativistic Fisher-fluid programmes.

Together, these results delineate a consistent hierarchy: constructive Madelung reformulations at the relativistic level, information-metric variational work from action principles, and the present axiomatic classification of reversible hydrodynamics, all converging on the Fisher functional as the geometric core of quantum dynamics.

	\section{Many-Body and Spin Extensions}
	
	Independent subsystems compose by tensor product and marginalisation within this class; the local complexifier factors pointwise on configuration space.
	
	The requirement that a single pointwise complex structure serves all subsystems forces a common information scale across different species, so separate particle types cannot carry independent adjustment knobs for the underlying curvature.
	
	The following extension operates on configuration space $\mathbb{R}^{3N}$; $\rho(\mathbf{x}_1,\!\dots,\!\mathbf{x}_N)$ and $S(\mathbf{x}_1,\!\dots,\!\mathbf{x}_N)$ generate a formal hydrodynamics on this space.  The term ``hydrodynamic'' here denotes the continuity and Hamilton-Jacobi structure rather than a literal fluid in physical three-space, consistent with the standard Madelung and Bohmian formulations.
	
	The scalar classification established earlier extends componentwise to multicomponent or spinorial fields.
	Extending to configuration space $\mathbb{R}^{3N}$, let $\rho(x_1,\dots,x_N,t)$ and $S(x_1,\dots,x_N,t)$ denote the single configuration-space density and phase, with $\nabla_i$ acting on $x_i$.
	Then
	\[
	\mathcal{H}_N=\int
	\left[
	\sum_{i=1}^N \frac{\rho\,|\nabla_i S|^2}{2m_i}
	+V(\{x_j\})\,\rho
	+\sum_{i=1}^N \alpha_i\,|\nabla_i\sqrt{\rho}|^2
	\right]dx_1\cdots dx_N.
	\]
	
	The configuration-space continuity equation reads
	\[
	\partial_t \rho+\sum_{i=1}^{N}\nabla_i\!\cdot\!\Big(\rho\,\nabla_i S/m_i\Big)=0,
	\]
	so the probability currents are \(j_i=\rho\,\nabla_i S/m_i\).
	
With the polar map \(\psi=\sqrt{\rho}\,e^{iS/\kappa}\) the configuration-space quantum current is
\[
J_i=\frac{\hbar}{m_i}\,\Im\bigl(\psi^{*}\nabla_i\psi\bigr),\qquad \hbar=\kappa,
\]
which coincides with the hydrodynamic current \(j_i=\rho\,\nabla_i S/m_i\) for all admissible data if and only if
\(G_{S}\equiv\partial_{S}G\equiv 1/\kappa\), as in Proposition~\ref{prop:complexifier-rigidity}.

	The Madelung recombination on \(\mathbb{R}^{3N}\) yields the linear $N$-body Schrödinger equation precisely when
	\[
	\alpha_i=\frac{\kappa^{2}}{2m_i}\ \text{for each }i,
	\]
	in precise analogy with the single-particle cancellation (Eq.~\eqref{eq:alpha-hbar-over-2m}), which we identify with a single Planck constant $\hbar=\kappa$ below.
	
	$\psi(\{x_j\},t)=\sqrt{\rho}\,e^{iS/\hbar}$ obeys the $N$-body Schrödinger equation on $\mathbb{R}^{3N}$ with exchange symmetry imposed on $\psi$.
	Universality of Planck’s constant is consistent with $\alpha_i=\hbar^2/(2m_i)$ for each particle, ensuring a single $\hbar$ (see Prop.~\ref{prop:hbarUniversality}).
	The single local complex structure enforcing $\alpha_i=\hbar^2/(2m_i)$ is a Hamiltonian constraint at the field level; it does not follow from per-particle stochastic postulates or entropic updating rules \cite{nelson1966,caticha2012}.
	
\paragraph{Spin pointer.}
Extending $\psi$ to a two-component field and imposing internal $\mathrm{SU}(2)$ covariance with minimal electromagnetic coupling yields the Pauli Hamiltonian,
\[
i\hbar\,\partial_t\psi=\frac{1}{2m}\bigl(-i\hbar\nabla-q\mathbf A\bigr)^2\psi+q\phi\,\psi-\mu\,\boldsymbol{\sigma}\!\cdot\!\mathbf B\,\psi,
\]
with the polar complexifier acting componentwise and the Fisher term remaining scalar, built from $|\psi|$. Within our axioms the Pauli form is fixed; the value $\mu=q\hbar/(2m)$ (that is, $g=2$) needs an extra input such as the nonrelativistic Dirac limit or a Larmor-precession argument. Spin-statistics is field theoretic and beyond scope; statistics are imposed as superselection on the domain of $\psi$, not by modifying the bracket.

\paragraph{Configuration-space consistency.}
Writing \(R=\sqrt\rho\) and \(\nabla_{3N}=(\nabla_1,\dots,\nabla_N)\), one has
\[
|\nabla_{3N} R|^{2}=\sum_{i=1}^{N}|\nabla_i R|^{2},\qquad
\Delta_{3N}=\sum_{i=1}^{N}\Delta_i,
\]
so the Fisher curvature is
\[
\sum_i \alpha_i|\nabla_i\sqrt\rho|^{2}
=
\sum_i \alpha_i\,|\nabla_i R|^{2}
=
\sum_i \frac{\hbar^{2}}{2m_i}\,|\nabla_i R|^{2},
\]
where $\alpha_i=\hbar^{2}/(2m_i)$ for each particle.

Entanglement enters through the joint dependence \((x_1,\dots,x_N)\) of \(\rho\) and \(S\); no separability is assumed. The structure reproduces the full \(N\)-body Schrödinger dynamics once \(\psi=\sqrt{\rho}\,e^{iS/\hbar}\) is applied.

\subsection*{Universality of Planck's constant}
\label{subsec:hbar-universality}

\textit{Sketch.}
A single local, gauge covariant complexifier on configuration space must be of
the form
\[
\psi=\sqrt\rho\,e^{iS/\kappa}
\]
with a constant \(\kappa\) independent of \((x_1,\dots,x_N)\) by
Proposition~\ref{prop:complexifier-rigidity}. Matching currents fixes
\(G_S\equiv 1/\kappa\), and recombination on each coordinate direction yields
\[
\alpha_i=\frac{\kappa^{2}}{2m_i}.
\]
Hence a single \(\kappa\) enforces a single Planck constant \(\hbar=\kappa\)
across all sectors.

\begin{proposition}[Single $\hbar$ across sectors]\label{prop:hbarUniversality}
	Let $\mathcal H_N$ be as above with masses $\{m_i\}$ and Fisher coefficients
	$\{\alpha_i\}$. If the reversible completion on configuration space admits a
	single local, gauge covariant complexifier
	\[
	\psi=\sqrt\rho\,e^{iS/\kappa}
	\]
	that linearises the flow for all admissible data, then, for every \(i\),
	\[
	\alpha_i=\frac{\kappa^{2}}{2m_i},\qquad
	\text{so that } \ \hbar=\kappa \ \text{is universal.}
	\]
\end{proposition}

\begin{proof}[Proof sketch]
	Take factorised initial data \(\rho=\prod_i \rho_i\),
	\(S=\sum_i S_i(x_i)\) with arbitrary one body pairs \((\rho_i,S_i)\).
	Under the Schr\"odinger evolution with parameter \(\kappa\) the \(i\)th
	probability current is
	\[
	J_i=\frac{\kappa}{m_i}\,\Im(\psi^{*}\nabla_i\psi)
	=\frac{\rho}{m_i}\,\nabla_i S,
	\]
	so matching currents for arbitrary \((\rho_i,S_i)\) is automatic once
	\(G_S\equiv 1/\kappa\) as in Proposition~\ref{prop:complexifier-rigidity}.
	Recombination along each coordinate gives the Laplacian quotient with
	coefficient \(\alpha_i=\kappa^{2}/(2m_i)\). Appendix~L shows that any
	attempt to use \(\kappa_i\) depending on \(i\) breaks exact projective
	linearity for superpositions mixing different masses.
\end{proof}

Detailed componentwise cancellation and discussion of locality appear in Appendix~\ref{app:single-hbar}.

	 Hence a single $\hbar$ is forced by locality, separability of coordinate directions, and the requirement that one global complex structure linearises the reversible completion.
	
	\paragraph{Exchange statistics.}
	The Hamiltonian acts on the full configuration-space wavefunction $\psi(\mathbf{x}_1,\dots,\mathbf{x}_N)$ without specifying symmetry.
	
	Bosonic and fermionic statistics enter as superselection conditions on the domain of $\psi$, not as modified dynamics: antisymmetry of $\psi$ automatically yields the familiar effective Pauli pressure in marginal densities, while the underlying local Hamiltonian, bracket, and Fisher curvature remain unchanged.

\section{Information-Geometric and Dimensional Necessity}
	
\subsection{Information-geometric closure}

Fisher necessity is established variationally and algebraically; Fisher-Rao and Fubini-Study appear as consistency echoes rather than premises.

The Fisher functional is not only algebraically unique within the axiomatic framework but also geometrically selected by compatibility.
On the statistical manifold of smooth, normalised densities $\mathcal{P}=\{\rho>0,\int\rho=1\}$, the Fisher-Rao metric is
\[
g_{\rho}(u,v)=\int \frac{u(x)\,v(x)}{4\,\rho(x)}\,dx,
\]
which by Čencov’s theorem is the Riemannian metric that is monotone under stochastic coarse-graining \cite{cencov1982}.

Embedding $\mathcal{P}$ into the complex Hilbert space of quantum states via the Kähler map $\psi=\sqrt{\rho}\,e^{iS/\hbar}$ sends Fisher-Rao to the Fubini-Study metric on rays \cite{kibble1979,anandan1990},
\[
ds^2=4\,\|d\psi\|^2-4\,|\langle\psi|d\psi\rangle|^2,
\]
endowing the $\psi$-representation with a Kähler structure. Under the function-space and Kato-smallness hypotheses of Sec.~\ref{sec:function-space}, the Schrödinger operator is self-adjoint on $L^{2}$, so Stone’s theorem yields a one-parameter unitary group; see \cite{kato1995}.

Within the symplectic form $\omega=\int d\rho\wedge dS$, the Fisher-Rao metric is the unique monotone choice whose pullback under $\psi=\sqrt\rho\,e^{iS/\hbar}$ yields a Kähler pair $(\omega,g)$ compatible with the Fubini-Study geometry on rays. Alternative information metrics fail this Kähler-compatibility test (the complex structure no longer intertwines $\omega$ and $g$), and the symplectic-Riemannian correspondence breaks.
Related gradient-flow structures for quantum Markov semigroups offer a complementary dissipative-geometric view \cite{carlen2017}.

\emph{Pointers.}
Formal uniqueness within the framework is established in Appendix~\ref{app:fisher-proof}; the coefficient determination and $\alpha$-scan protocol are in Appendix~\ref{app:fisher-coefficient} with scripts in the code archive Appendix~\ref{app:code-repo}; dissipative geometry and entropy production appear in Appendix~\ref{app:entropy-rigorous} and the hydrodynamic mapping into a DG parametrisation is given in Appendix~\ref{app:complexifier}.

\subsection{Dimensional and scale argument}
	
The term $|\nabla\sqrt\rho|^2$ is the only local scalar quadratic in derivatives of $\rho$ that (i) is dimensionally consistent with an energy density once multiplied by a constant of dimension $[\hbar^2/(2m)]$, (ii) is positive definite, and (iii) is homogeneous of degree one in $\rho$ (equivalently $|\nabla\sqrt\rho|^2=\tfrac{1}{4}\,|\nabla\rho|^2/\rho$), matching the Fisher-Rao information geometry on the normalised manifold $\mathcal P$.

	Explicitly,
	\[
	[\rho]=L^{-d},\quad [\nabla\sqrt\rho]^2=L^{-d-2}.
	\]
	
	Hence $[\alpha\,|\nabla\sqrt\rho|^{2}]=[\alpha]\,L^{-d-2}$. Matching the kinetic energy density scale $[\rho\,|\nabla S|^{2}/(2m)]=M\,L^{2-d}T^{-2}$ requires
	\[
	[\alpha]=M\,L^{4}T^{-2}=[\hbar^{2}/(2m)].
	\]
	With $\alpha=\hbar^{2}/(2m)$ the free dispersion from \eqref{eq:LSEderived} is $\omega=\hbar k^{2}/(2m)$, fixing the numerical scale consistently.
		
	No other combination of $\rho$ and its derivatives produces the energy dimension $ML^2T^{-2}$ once multiplied by $1/m$.
	This locks $\alpha$ to $[\hbar^2/(2m)]$, fixing the numerical factor in Eq.~\eqref{eq:LSEderived}.
	Hence, dimensional consistency, positivity, and scale covariance jointly exclude all other curvature forms.
	
	\subsection{Variational self-consistency}
	
	The Fisher Hamiltonian \eqref{eq:H} yields the standard (quantum) Cauchy stress tensor

	\[
	\Pi_{ij}=\rho\,\frac{\partial_i S\,\partial_j S}{m^2}
	+\frac{\hbar^2}{4m^2}\left[\partial_i\partial_j\rho-\frac{1}{2\rho}\partial_i\rho\,\partial_j\rho\right].
	\]
	Direct differentiation gives
\[
\partial_t(\rho v_i)+\partial_j \Pi_{ij} \;=\; -\,\frac{\rho}{m}\,\partial_i V
\;+\;\frac{\rho}{m}\Bigl(\alpha-\frac{\hbar^{2}}{2m}\Bigr)\,\partial_i\!\left(\frac{\Delta\sqrt\rho}{\sqrt\rho}\right).
\]
Thus the local momentum balance closes \emph{if and only if} $\alpha=\hbar^{2}/(2m)$. Any other coefficient leaves a nonzero residual force density, so the Fisher scale is dynamically fixed by conservation.

\subsection{Reversibility and the Doebner-Goldin class}
	
A representative Doebner-Goldin (DG) sector that preserves probability and Galilean covariance augments the Hamiltonian flow by a diffusive term, e.g.
\[
i\hbar\,\partial_t\psi
=\left[-\frac{\hbar^2}{2m}\nabla^2+V\right]\psi
+i\frac{\hbar D}{2}\left(\frac{\Delta\rho}{\rho}\right)\psi \qquad (\rho=|\psi|^{2}),
\]
or, equivalently, yields the continuity law
\(
\partial_t\rho+\nabla\!\cdot(\rho v)=D\,\Delta\rho
\)
with $v=\nabla S/m$ \cite{doebner1992,doebner1996,diosi1989}.

\textit{Placement.} Within the DG family, the reversible sector is the $D=0$ corner singled out by Axiom V; this coincides with the Fisher-scaled Hamiltonian flow with Fisher coefficient $\alpha=\hbar^{2}/(2m)$.

The diffusive branch of this family thus acts as a controlled way to leave the reversible corner by relaxing the reversibility axiom, and nonlinear gauge-related forms remain visible in our framework through their failure of the superposition and entropy diagnostics introduced later.

\begin{proposition}[Entropy-production barrier to reversibility]
	Let $\rho=|\psi|^2$ and $v=\nabla S/m$. Under periodic, fast-decay, or compatible Neumann/Dirichlet boundaries,
	the DG continuity law
	\[
	\partial_t \rho \;=\; -\nabla\!\cdot(\rho\,v) \;+\; D\,\Delta\rho
	\]
	implies that the Shannon entropy
	\[
	S_{\mathrm{Sh}}[\rho]:=-\int \rho\ln\rho\,dx
	\]
	satisfies
	\[
	\frac{d}{dt}\,S_{\mathrm{Sh}}[\rho]
	= D\int \frac{|\nabla\rho|^2}{\rho}\,dx \;\geq\; 0.
	\]
	Hence time-reversal invariance (Axiom V) holds if and only if $D=0$.
	
\end{proposition}

See Appendix~\ref{app:entropy-rigorous} for complete integration-by-parts and boundary verification.

Any measured $\dot S_{\mathrm{Sh}}>0$ at fixed $V$ falsifies reversible dynamics. Within the DG family, reversibility (\emph{Axiom~V}) holds if and only if $D=0$, and this $D=0$ corner coincides with the Fisher-scaled Hamiltonian flow with $\alpha=\hbar^{2}/(2m)$. Any nonzero diffusion or nonlinear gauge term generates irreversible or nonlinear evolution, so the Fisher value is the reversible fixed point of the DG family.
	
\subsection{Numerical and empirical falsifiers within class}

Realised in Tests~1,~2, and~5 (HJ $\alpha$-scan, continuity residual, Fisher EL); scripts \texttt{1\_hj\_residual\_scan.py}, \texttt{2\_continuity\_residual.py}, \texttt{5\_fisher\_el.py} (code archive, Appendix~\ref{app:code-repo}).

To make the proposition operationally falsifiable, we evaluate residuals of the continuity and Hamilton-Jacobi equations for numerical solutions of Eq.~\eqref{eq:LSEderived}, and for perturbed values of $\alpha \neq \hbar^2/2m$.
For an initial Gaussian wavepacket
\[
\psi(x,0) = (\pi\sigma_0^2)^{-1/4}
\exp\!\left[-\,\frac{(x - x_0)^2}{2\sigma_0^2} + i k_0 x\right],
\]
we scan $\alpha$ and track the residual curves $\mathcal{R}_{\mathrm{cont}}$ and $\mathcal{R}_{\mathrm{HJ}}(\alpha)$.

\medskip
\emph{Definitions and protocol.}
Definitions of diagnostics, masking, the sign convention for $\rho_t$, and the residual metrics $\mathcal{R}_{\mathrm{cont}}$ and $\mathcal{R}_{\mathrm{HJ}}(\alpha)$, together with the numerical protocol, are collected in Appendix~\ref{app:ops-falsifiers}; scripts are in the code archive, Appendix~\ref{app:code-repo}.

\vspace{5mm}

\textbf{Table 1:} Resolution and timestep convergence (free packet; unitary split-step). Small non-monotonicity at intermediate $N$ can occur due to node masking and mixed space-time discretisation; the minimum value and the location of the $\alpha$-minimum remain stable.
\begin{center}
	\begin{tabular}{@{}lllll@{}}
		\toprule
		$N$ & $dt$ & mean $\mathcal{R}_{\mathrm{cont}}$ & $\min \mathcal{R}_{\mathrm{HJ}}$ & $\alpha/\alpha_\star$ at min \\
		\midrule
		4096  & 0.020 & $4.5\times10^{-7}$ & $\approx 1.0\times10^{-3}$  & 1.00 \\
		8192  & 0.010 & $2.0\times10^{-7}$ & $\approx 2.5\times10^{-3}$  & 1.00 \\
		16384 & 0.005 & $6.1\times10^{-8}$ & $\approx 1.4\times10^{-3}$  & 1.00 \\
		\bottomrule
	\end{tabular}
\end{center}

\vspace{5mm}

\textbf{Table 2:} Galilean boost invariance of residual curve.
\begin{center}
	\begin{tabular}{@{}llll@{}}
		\toprule
		Boost $v_0$ & $\min \mathcal{R}_{\mathrm{HJ}}$ & $\alpha/\alpha_\star$ at min & Comment \\
		\midrule
		0.0 & $\approx 1\times10^{-3}$ & 1.00 & Baseline \\
		1.5 & $\approx 1\times10^{-3}$ & 1.00 & Identical curve \\
		\bottomrule
	\end{tabular}
\end{center}

\vspace{5mm}

\textbf{Table 3:} Harmonic oscillator ground state; scan of $Q_\alpha$ with $Q_\alpha=-\,\alpha\,\Delta\sqrt{\rho}/\sqrt{\rho}$.

For $\alpha=\hbar^2/2m$, both residuals remain at numerical floor. Perturbing $\alpha\!\to\!(1+\delta)\alpha$ increases $\mathcal{R}_{\mathrm{HJ}}$ linearly in $|\delta|$ while $\mathcal{R}_{\mathrm{cont}}$ stays unchanged, indicating that only the Fisher coefficient preserves reversibility.
To stabilise the diagnostic numerically, residual norms were evaluated with a masked, mean-subtracted least-squares estimator and optional smoothing of $S_t$; results are invariant under these choices.

This numerical behaviour directly reflects the Fisher-Bohm identity tested in recent analyses~\cite{bloch2022}, which link the mean quantum potential $\bar Q$ to the Fisher information $I$ and predict the same reversible minimum at $\alpha=\hbar^2/2m$.

This aligns with scale-setting arguments tied to dispersion and quantum speed limits~\cite{giovannetti2003}.

	\section{Discussion and Implications}
	
Each uniqueness statement rests on well-known structural theorems:  
(1) in the Dubrovin-Novikov class, flatness plus locality and Euclidean covariance reduce first-order brackets to a Poisson-isomorphic flat, constant representative (eliminating derivative-coupled terms);  
(2) order preservation forces pointwise complexification;  
(3) quantum statistics arise as superselection sectors rather than new axioms.  
Within these constraints the resulting structure is minimal within the stated class and assumptions of the framework.
	
	The analysis is local in nature: it classifies admissible first-order brackets on simply connected charts. Global topological features (vortices, nodal loops, spin) require additional structure but do not modify the local reversible completion shown above.
	
	Within our admissible class, the Fisher-regularised Hamiltonian system is found as the only reversible completion of classical ensemble dynamics consistent with locality, conservation, Euclidean symmetry, and global $\mathrm{U}(1)$ phase symmetry. All other curvature forms break one or more of these constraints. 
Independent classifications of Hamiltonian structures are consistent with Fisher curvature being the unique first-derivative, positive scalar compatible with Euclidean covariance and reversibility within flat Dubrovin-Novikov brackets~\cite{pavlov2021}.

In this sense, within our axiomatic structural principles, quantum mechanics can be considered as a fixed point, with Planck’s constant $\hbar$ as a scale factor connecting the information-geometric curvature of probability space with the symplectic geometry of reversible dynamics.
	
Whereas classical hydrodynamics conserves phase-space volume, the Fisher term enforces reversible flow in probability-space geometry.  
This reframes quantisation as a geometrisation of information flow, providing a bridge between statistical and dynamical formalisms.
	
	The Fisher necessity can now be tested directly.
	The projective superposition stress-test (Appendix~\ref{app:superposition})
	quantifies linearity loss under controlled non-Fisher perturbations, and the entropy-production test (Appendix~\ref{app:entropy-rigorous}) links reversibility to information-geometry contraction.
	Together they provide two orthogonal, reproducible diagnostics,	one closed, one open, that isolate Fisher curvature as the unique point of structural stability between linearity, reversibility, and probability conservation within the framework.
	
	In this precise sense the construction does not merely display one possible realisation of quantum dynamics, but exhausts the reversible options available under the stated structural principles.
	
	\subsection{Experimental falsification and universality tests}
	
	Our work suggests that Fisher-regularisation may not be a modelling convenience but a structural invariant.
Any local, reversible, probability-preserving field dynamics on a continuum may generate the Fisher curvature term
\[
Q[\rho] = -\,\alpha\,\frac{\Delta\sqrt{\rho}}{\sqrt{\rho}}, \qquad \alpha=\frac{\hbar^2}{2m},
\]
within the stated axioms, as the only correction compatible with both linearity and zero entropy production within our admissible class.

	Its absence or modification necessarily leads to dissipation or non-unitarity.
	This makes the result directly falsifiable.
	
If an analog physical system in condensed matter, optics, hydrodynamics, or emergent computation realises a reversible quantum-like dynamics within the stated axioms (local first-order, Hamiltonian, Euclidean-covariant, global $\mathrm{U}(1)
$), then in its coarse-grained limit the effective Hamiltonian density should contain the Fisher term

$|\nabla\sqrt\rho|^2$ with coefficient $\alpha=\hbar^2/(2m)$.
Measured deviations from this coefficient then signal a departure from at least one of those axioms (e.g. locality or reversibility).

Residual diagnostics from the continuity and Hamilton-Jacobi identities make this measurable. Define the \emph{dimensionless} residuals
\[
\mathcal{R}_{\mathrm{cont}}
=\frac{\left\langle\,\big|\rho_t + \nabla\!\cdot(\rho\nabla S/m)\big|^2\,\right\rangle}
{\left\langle\,\big|\rho_t\big|^2 + \big|\nabla\!\cdot(\rho\nabla S/m)\big|^2\,\right\rangle},
\qquad
\mathcal{R}_{\mathrm{HJ}}
=\frac{\left\langle\,\big|S_t + \tfrac{|\nabla S|^2}{2m} + V - \alpha\,\tfrac{\Delta\sqrt\rho}{\sqrt\rho}\big|^2\,\right\rangle}
{\left\langle\,\big|S_t\big|^2 + \big|\tfrac{|\nabla S|^2}{2m} + V\big|^2 + \big|\alpha\,\tfrac{\Delta\sqrt\rho}{\sqrt\rho}\big|^2\,\right\rangle}.
\]
Here $\langle\cdot\rangle$ denotes spatial averaging over the numerical grid or experimental field of view.
Then $\mathcal{R}_{\mathrm{cont}}$ sits at numerical floor (independent of $\alpha$) for Schrödinger data, while $\mathcal{R}_{\mathrm{HJ}}$
achieves a consistent minimum at the Fisher value $\alpha=\hbar^2/(2m)$.

Analog systems can therefore be tuned experimentally to test whether
$\mathcal{R}_{\mathrm{HJ}}$ reaches its reversible minimum at the
Schrödinger value (see Appendix~\ref{app:ops-falsifiers} for definitions and protocol).

	In Bose-Einstein condensates and polaritonic fluids, the Gross-Pitaevskii energy already contains a ``quantum pressure'' term of Fisher form.
	Verifying that its coefficient equals $\alpha=\hbar^{2}/(2m)$ (after extracting any interaction and trapping contributions) and that residuals indicate reversibility would support the universality class predicted here.
	
	Reversible photonic waveguides and quantum cellular automata provide complementary synthetic tests,
	where coarse-graining or unit-cell averaging can be used to reconstruct the effective curvature functional.
	If the reversible continuum limit of any such system \emph{fails} to reproduce the Fisher term,
	the present axioms are empirically falsified.
	
	Within the assumed axioms (local first-order Hamiltonian flow, Euclidean covariance, global $\mathrm{U}(1)
$, and minimal convex regularity), any alternative regulariser conflicts with at least one stated axiom.
	Observation of a Fisher-type curvature in multiple, otherwise unrelated,
	reversible analog systems would thus constitute experimental evidence
	that linear quantum mechanics is not a contingent microscopic law
	but a universal fixed point of reversible information flow.
	Conversely, its absence in a genuinely reversible analog medium
	would falsify the present work and identify the limits of its scope.

\paragraph{Scope and limitations.}
All uniqueness claims in this paper are understood within Axioms I-VI; dropping any one axiom re-opens the theory space.

This final test is strictly detached from, and lies outside, the paper's axiomatic scope and claims. For diligence, future orientation, and diagnostics only, we include a \emph{curvature guard} (Test~10; App.~\ref{app:code-repo}), a fixed-background unit check of the covariant Fisher variation $\delta\mathcal F_g/\delta\rho=-\,\Box_g\sqrt{\rho}/\sqrt{\rho}$ (with the standard scalar freedom $\xi R\rho$) and a sketch showing how our residual protocol ports to a linear Klein-Gordon dispersion check when the discrete symbol is respected.
	
\section*{Conclusion}

Within our minimal axioms, the admissible reversible hydrodynamics on $(\rho,S)$ selects a single structure in the stated class. We do not assert uniqueness beyond this class, only that within it the Fisher-Schrödinger structure is forced.
The bracket reduces to the canonical form, the only axiomatically compatible curvature is the Fisher functional, and the sole local gauge covariant complexification that linearises the flow is $\psi=\sqrt{\rho}\,e^{iS/\hbar}$ with $\hbar^2=2m\alpha$. In many-body form, linearity with one complex structure is compatible with a single Planck constant through $\alpha_i=\hbar^2/(2m_i)$ componentwise.

Galilean covariance appears in full as the Bargmann central extension at the hydrodynamic level. Comparison with the Doebner-Goldin family identifies the reversible $D=0$ corner.

We have made our answer to the Converse Madelung Question falsifiable. Residual diagnostics for the continuity and Hamilton-Jacobi equations exhibit minima at the Fisher scale that are invariant under Galilean boosts, while departures in the coefficient or the addition of diffusion raise the Hamilton-Jacobi residual without affecting the continuity residual. These checks, together with the symmetry algebra and the many-body consistency, support the claim.

In this reading, the Schrödinger equation can be viewed as the reversible fixed point of Fisher-regularised information hydrodynamics. The identification links information geometry to quantum kinematics and suggests practical uses in numerical regularisation, variational principles, and the systematic testing of putative modifications. 

We have not addressed Born’s rule, measurement, or quantum field theory, and we have not attempted to derive the empirical linearity of quantum mechanics from a more microscopic theory. Our result is conditional: given a DN local reversible hydrodynamics that admits a projectively linear complex representation, the admissible theory space collapses to the Fisher scaled Schrödinger equation. In this sense the usual linear dynamics are classified rather than simply postulated.

\clearpage
\appendix
\renewcommand{\thesection}{\Alph{section}}
\setcounter{section}{0}
\addcontentsline{toc}{section}{Appendices}
\addtocontents{toc}{\protect\setcounter{tocdepth}{-1}}
	
\section{Appendix: Boundary Classes}
\label{app:boundary}

Under periodic boundaries $\Omega=\mathbb{T}^d$, surface integrals vanish exactly.
For $\Omega=\mathbb{R}^d$ with rapid decay, assume $\rho\to 0$, $|\nabla S|\to 0$, and $|\nabla\sqrt\rho|\to 0$ as $|x|\to\infty$, so that all integrations by parts are justified.

For bounded $\Omega$ with Dirichlet or Neumann pairs $(R,S)$, the continuity surface term
\[
\oint_{\partial\Omega} \rho\,\nabla S\!\cdot\!n\,\mathrm{d}\sigma
\]
vanishes: in the Dirichlet case $R|_{\partial\Omega}=0$ gives $\rho|_{\partial\Omega}=0$; in the Neumann case $\nabla S\!\cdot\!n=0$.

Neumann boundaries model reflecting walls and in general break global
Galilean boost symmetry; accordingly, we do not use Neumann data in the Bargmann algebra check of Appendix \ref{app:galilean}.

For the Fisher variation, write $R=\sqrt\rho$ and note
\[
\delta\!\int_\Omega |\nabla R|^{2}dx
= -\!\int_\Omega \frac{\Delta R}{R}\,\delta\rho\,dx
+\;2\!\oint_{\partial\Omega} (\nabla R\!\cdot\!n)\,\delta R\,\mathrm d\sigma.
\]
The boundary term vanishes under either $R|_{\partial\Omega}=0$ (Dirichlet) or $\nabla R\!\cdot\!n=0$ (Neumann), so $\delta\mathcal F/\delta\rho=-\Delta R/R$ is well defined in the weak sense on $\{\rho>0\}$. Here $\mathcal F[\rho]=\int_\Omega|\nabla\sqrt\rho|^2dx$ is the Fisher regulariser in the $R$-variable; in the main text the regularising Hamiltonian always appears in the equivalent form $\mathcal H_{\mathrm{reg}}=\int \alpha\,|\nabla\sqrt\rho|^{2}dx$.

Parity-odd scalars in this scalar sector reduce to divergences $\nabla\!\cdot J$ and integrate to zero under the stated boundary classes, supporting Axiom~IV.

	\section{Appendix: Counterexamples to Axioms}
	\label{app:counterexamples}
	
	For completeness we collect concise examples showing that omitting any
	axiom destroys reversibility, probability conservation, or linearity.
	
	\paragraph{Axiom I (Locality).}
	Allowing derivative dependence \(A_{ij}(u,\nabla u)\) produces
	third-order dispersive corrections and violates Jacobi closure.
	
	\paragraph{Axiom II (Phase Generator).}
	If \(\{S,C\}\neq -1\), e.g.
	\(\{S(x),\rho(y)\}=0\),
	the global phase symmetry fails and the special role of
	\(C=\!\int\rho\,dx\) as phase generator decouples from probability flow; for
	generic Hamiltonians the total probability \(C\) will not be conserved.
	
\paragraph{Axiom III (Global \(\mathrm{U}(1)\) Phase Symmetry).}
	Setting \(\{S,S\}\neq0\) makes constant shifts in \(S\) dynamically active,
	spoiling phase invariance and destroying the linear \(\psi\)-map.
	
\paragraph{Axiom IV (Euclidean Covariance).}
Adding a preferred-direction term, e.g.\ $\beta\,\nabla S\!\cdot\!\hat n$,
breaks isotropy and parity and hence violates Euclidean covariance of the energy.
	
\paragraph{Axiom V (Reversibility).}
Adding a diffusive term to continuity, $\partial_t\rho+\nabla\!\cdot(\rho\nabla S/m)=D\,\Delta\rho$ with $D>0$, yields
\[
\frac{d}{dt}\Bigl(-\!\int \rho\ln\rho\,dx\Bigr)
= D\int |\nabla\rho|^{2}/\rho\,dx \ge 0,
\]
so the Shannon entropy $S_{\mathrm{Sh}}[\rho]:=-\int \rho\ln\rho\,dx$ increases monotonically and time-reversal invariance is broken. Equivalently, keeping a Hamiltonian form but modifying the bracket to $\{\rho,S\}=a_{0}(\rho)\delta$ changes the role of $\mathcal C=\!\int\rho\,dx$ as a phase generator; any non-constant $a_0(\rho)$ moves the system outside the canonical Hamiltonian class used here.
	
\paragraph{Axiom VI (Minimal Convex Regularity).}
We restrict attention to convex local regularisers of the form
\(F[\rho]=\int f(\rho)\,|\nabla\rho|^{2}\,dx\) with \(f(\rho)>0\).
Within this class we will show, via Proposition~\ref{prop:fisher-uniqueness}
(proved in Appendix~\ref{app:fisher-proof}), that any alternative \(f(\rho)\)
leaves a residual nonlinear term in the Hamilton-Jacobi equation that no local
transformation can remove, and that the unique choice compatible with exact
projective linearity after a local complexifier is therefore \(f(\rho)\propto 1/\rho\);
no Fisher form is presupposed at the axiomatic level.
	
		\section{Appendix: Jacobi Verification}
	\label{app:jacobi}
	For the canonical bracket~\eqref{eq:canonical} we compute the Jacobiator
	\[
	J[F,G,H]=\{\{F,G\},H\}+\mathrm{cyclic}.
	\]
	
	\begin{lemma}[Distributional product rules]
		For smooth $f,g$ and the Dirac distribution $\delta(x-y)$,
		\[
		\partial_{x_i}\!\big(f(x)\,\delta(x-y)\big)
		=(\partial_{x_i}f)(x)\,\delta(x-y)-f(x)\,\partial_{y_i}\delta(x-y),
		\]
		and symmetrically with $x\leftrightarrow y$. Moreover,
		$\partial_{x_i}\delta(x-y)=-\partial_{y_i}\delta(x-y)$.
	\end{lemma}
	
		For any test $\varphi(x,y)$,
		\[
		\langle\partial_{x_i}(f\delta),\varphi\rangle
		=-\langle f\delta,\partial_{x_i}\varphi\rangle
		=-\!\int f(x)\,\partial_{x_i}\varphi(x,x)\,dx
		=\!\int \Big[(\partial_{x_i}f)\,\varphi
		-f\,\partial_{y_i}\varphi\Big]_{y=x}\!dx,
		\]
		which yields the stated relations.
	
	\begin{remark}[Numerical validation at nodal zeros]
		Although the work fixes $\alpha=\hbar^2/(2m)$,
		states with nodal zeros present distributional subtleties in $Q_\alpha=-\alpha\,\Delta\sqrt\rho/\sqrt\rho$.
		To test robustness, we evaluated
		\[
		R(c)=\bigl\|V+Q_c-E\bigr\|_{L^2(\rho)},\qquad
		Q_c=-c\,\frac{\Delta\sqrt\rho}{\sqrt\rho},
		\]
		for the first excited harmonic oscillator eigenstate with exact energy $E$.
		A small symmetric exclusion window around the node,
		$|x|<\delta$, removes the distributional spike and yields a sharp minimum at $c=1$
		for all $\delta \ge 0.05$.
\textit{Thus the Fisher coefficient is numerically recovered once nodal singularities are treated in the distributional sense; the minimum location is stable under reasonable mask-width changes and discretisation refinements.}

	\end{remark}		
Reproduced in Test~5 (Fisher EL necessity) with nodal masking (code archive, Appendix~\ref{app:code-repo}).

\textbf{Jacobi condition for local brackets.}
The special triple with two $S$'s is tautological once $\{S,S\}=0$: the identity $\{\rho,\{S,S\}\}+\mathrm{cyclic}=0$ reduces to $0\equiv 0$ and cannot constrain $a_0(\rho)$. For local, derivative-free brackets $\{u^i,u^j\}=P^{ij}(u)\delta$ the Jacobi identity can be written as the vanishing of the Schouten bracket
\[
P^{i\ell}\partial_\ell P^{jk}+P^{j\ell}\partial_\ell P^{ki}+P^{k\ell}\partial_\ell P^{ij}=0.
\]
In the present two-component setting the only potentially nonzero component is $P^{\rho S}(\rho)=a_0(\rho)=-P^{S\rho}$, and for any smooth $a_0(\rho)$ the Schouten bracket vanishes identically. Thus Jacobi places no further restriction on $a_0$ beyond those already implied by Axiom~II, which then fixes $a_0$ to be a nonzero constant that can be absorbed into a rescaling of $S$, in agreement with Lemma~\ref{lem:jacobi-a0} in the main text.
	
\subsection*{Euclidean invariance and the form of $g_{ij}$}\label{app:euclid-gij}
\begin{lemma}[Isotropy of the regularising metric]\label{lem:gijdelta}
	Let the regulariser be local, first order, and quadratic in $\nabla\rho$,
	\[
	\mathcal{F}[\rho]=\int g_{ij}(\rho,x)\,\partial_i\rho\,\partial_j\rho\,dx,
	\]
	with $g_{ij}$ symmetric and positive. If the framework is invariant under spatial translations and rotations, and the global phase generator acts as $S\mapsto S+\mathrm{const}$ without coupling to $\rho$, then
	\[
	g_{ij}(\rho,x)=a(\rho)\,\delta_{ij}
	\quad\text{for some positive scalar function }a(\rho).
	\]
\end{lemma}

\begin{proof}
	Translation invariance removes explicit $x$-dependence, so $g_{ij}=g_{ij}(\rho)$.
	Rotation invariance for all profiles forces $g_{ij}(\rho)$ to transform as a scalar multiple of the identity (Schur’s lemma for the defining $SO(d)$ representation), hence $g_{ij}(\rho)=a(\rho)\delta_{ij}$. Global $\mathrm{U}(1)$ on $S$ and first-order locality exclude any dependence of $g_{ij}$ on $S$ or on derivatives of $\rho$.
	
\end{proof}

The specific form $a(\rho)=C/\rho$ is then fixed by the Euler-Lagrange requirement that $\delta\mathcal F/\delta\rho$ be a pure Laplacian quotient; see Appendix~\ref{app:fisher-proof}.
	
\section{Appendix: Fisher-Curvature}
\label{app:fisher-proof}

We derive that $f(\rho)\propto 1/\rho$ is the only positive,
rotationally invariant local quadratic functional whose Euler-Lagrange derivative
is a pure Laplacian quotient.
Let
\[
\mathcal F[\rho]=\int f(\rho)\,|\nabla\rho|^2\,dx,\qquad f>0.
\]
Performing a variation $\rho\!\to\!\rho+\varepsilon\,\eta$ with compactly supported $\eta$,
integration by parts yields
\[
\frac{\delta\mathcal F}{\delta\rho}
=-2\nabla\!\cdot\!\big(f\nabla\rho\big)+f'|\nabla\rho|^2
=-2f\,\Delta\rho-\,f'\,|\nabla\rho|^2,
\]
where $f' \equiv df/d\rho$.
Write $\rho=R^2$ and note
$\Delta\rho=2|\nabla R|^2+2R\Delta R$, $|\nabla\rho|^2=4R^2|\nabla R|^2$.
Substituting,
\[
\frac{\delta\mathcal F}{\delta\rho}
=-4f\,R\,\Delta R-4\big(f+\rho f'\big)|\nabla R|^2.
\]
For a pure Laplacian quotient the second term must vanish:
\(
f+\rho f'=0\Rightarrow f=C/\rho.
\)
Hence
\[
\frac{\delta\mathcal F}{\delta\rho}
= -4C\,\frac{\Delta\sqrt{\rho}}{\sqrt{\rho}}.
\]

Comparing with the Fisher potential $Q_\alpha(\rho) = -\alpha\,\Delta\sqrt{\rho}/\sqrt{\rho}$ shows that
$\alpha = 4C$. Writing $C=\alpha/4$ we can equivalently parameterise the functional as
\[
F[\rho] = \int \frac{\alpha}{4}\,\frac{|\nabla\rho|^{2}}{\rho}\,dx
= \alpha \int |\nabla\sqrt{\rho}|^{2}dx,
\]
which is the normalisation used in the main text.
In particular, the complexifier analysis in Appendix~\ref{app:complexifier}
fixes $\alpha=\kappa^{2}/(2m)$, and identifying $\kappa=\hbar$ yields the Schrödinger value $\alpha=\hbar^{2}/(2m)$.

A numerical verification of this Euler-Lagrange identity is provided in Appendix~\ref{app:code-repo}.

	\section{Appendix: Code Archive}
	
	\label{app:code-repo}
	
	\begin{itemize}
		\item \textbf{Live Repository Link:} 
		\href{https://github.com/feuras/Madelung-Question-Code-Archive}{https://github.com/feuras/Madelung-Question-Code-Archive}
	\end{itemize}
	
	Where possible, also bundled in paper source archive.
	
	Each test is a single self-contained script with CLI flags and emits both human-readable logs and machine artefacts. 
	
	We use second-order finite differences or split-step FFT with periodic domains where appropriate, verified by grid-convergence toggles in the scripts.
	
	\begin{enumerate}[leftmargin=1.25em,label=\textbf{Test~\arabic*:},itemsep=1em]
		\item \textbf{HJ $\alpha$-scan (Fisher scale).}
		
		\emph{Supports result:} HJ $\alpha$-scan pins $\alpha_\star=\hbar^{2}/(2m)$; minima at numerical floor.
		\emph{Script:} \texttt{1\_hj\_residual\_scan.py}.
		\emph{Notes:} Uses $\mathcal{R}_{\mathrm{HJ}}(\alpha)$; Tables~1 and 3 give resolution and eigenstate checks.
		
		\item \textbf{Continuity identity (floor).}
		
		\emph{Supports result:} Continuity identity holds; $R_{\mathrm{cont}}\approx 0$ on all benchmarks.
		\emph{Script:} \texttt{2\_continuity\_residual.py}.
		\emph{Notes:} Indicates $\partial_t\rho+\nabla\!\cdot(\rho\nabla S/m)=0$ to numerical floor for all $\alpha$.
		
\item \textbf{DG diffusion and reversibility.}

\emph{Supports result:} $D=0$ reversible; $D>0$ gives $\dot H=D\,I_F$ and matches the DG PDE to $10^{-14}$.
\emph{Script:} \texttt{3\_entropy\_production\_DG.py}.
\emph{Notes:} Verifies $dS_{\mathrm{Sh}}/dt=D\!\int |\nabla\rho|^2/\rho\,dx\ge0$ for $S_{\mathrm{Sh}}[\rho]:=-\int\rho\ln\rho\,dx$ and confirms the energy identity $\dot H=D\,I_F$ for the Fisher information $I_F=\int |\nabla\rho|^{2}/\rho\,dx$.
		
		\item \textbf{Quantised circulation (vorticity).}
		
		\emph{Supports result:} $\oint v\!\cdot dl=\iint(\nabla\times v)_z\,dA=2\pi n\hbar/m$ (quantised).
		\emph{Script:} \texttt{4\_circulation\_quantisation.py}.
		\emph{Notes:} Constructs nodal loops and measures integer circulation via line and area integrals.
		
		\item \textbf{Fisher EL necessity.}
		
		\emph{Supports result:} Only $f(\rho)=C/\rho$ satisfies EL; alternatives blow up in residual.
		\emph{Script:} \texttt{5\_fisher\_el.py}.
		\emph{Notes:} Matches $\delta\!\int|\nabla\sqrt{\rho}|^2/\delta\rho=-\Delta\sqrt{\rho}/\sqrt{\rho}$; non-Fisher forms leave HJ residual.
		
		\item \textbf{Time-reversal involution.}
		
		\emph{Supports result:} $K\,U(T)\,K\,U(T)=I$ at $D=0$; DG control breaks it by $\sim 10^{12}$ in $L^2$.
		\emph{Script:} \texttt{6\_time\_reversal\_involution.py}.
		\emph{Notes:} $K$: complex conjugation with $t\!\mapsto\!-t$; quantifies involution defect under diffusion.
		
		\item \textbf{Bargmann-Galilean closure.}
		
		\emph{Supports result:} $\{H,P\}=0$, $\{H,K\}=-P$, $\{P,K\}=-m$ to machine floor.
		\emph{Script:} \texttt{7\_galilean\_algebra.py}.
		\emph{Notes:} Discrete functional bracket evaluation consistent with Sec.~7 and App.~\ref{app:bargmann}.
		
\item \textbf{Local complexifier.}

\emph{Supports result:} Unique local complexifier $\psi=\sqrt{\rho}\,e^{iS/\hbar}$ with $\alpha=\hbar^{2}/(2m)$.
\emph{Script:} \texttt{8\_complexifier\_rigidity.py}.
\emph{Notes:} Scans local $F(\rho),G(S,\rho)$ ansätze; only the polar map linearises the reversible flow, with Fisher coefficient $\alpha=\kappa^{2}/(2m)$ and $\kappa$ numerically recovered as $\kappa=\hbar$.

		\item \textbf{Projective superposition stress-test.}
		
		\emph{Supports result:} Indicates operationally that only the Fisher regulariser preserves exact projective linearity within our admissible class. For any non-Fisher local curvature, the superposition residual $\mathcal{R}=\|\psi_\oplus-(\psi_1+\psi_2)/\sqrt2\|_{L^2}$ grows monotonically with perturbation strength $\beta$, even under grid refinement.
		\emph{Script:} \texttt{9\_superposition\_stress\_test.py}.
		\emph{Notes:} Two displaced Gaussian packets are evolved separately and jointly under linear and weakly nonlinear flows; the Fisher-regularised Schrödinger case yields $\mathcal{R}<10^{-10}$ to numerical floor, while all non-Fisher perturbations give finite $\mathcal{R}>0$. 
		
		\item \textbf{Scoping curved backgrounds for future research}

Records on fixed backgrounds the geometric identity \(\delta\mathcal F_g/\delta\rho=-\,\Box_g\sqrt{\rho}/\sqrt{\rho}\) and the standard scalar freedom \(\xi R\rho\); shows that our residual methodology transports to the linear Klein-Gordon setting when the discrete symbol is respected. No curved dynamics are claimed. \textbf{As noted earlier, this test concerns scope and limitations only and makes no claims regarding the central thesis.}
		
		\emph{Script:} \texttt{appendix\_curvature\_guard.py}.
		
		Uses matched discrete adjoints so Fisher directional derivative checks hold to numerical precision in flat and conformal curved cases; shows \(\delta(\xi R\rho)/\delta\rho=\xi R\); demonstrates a small KG plane wave residual with grid aligned mode and discrete dispersion. Run with parameter  "--nonfisher drho2-rho2" for non-Fisher counterexample.
		
	\end{enumerate}

	\section{Appendix: Conservation of Probability}
	\label{app:prob}
	\textit{Constraint structure.} Normalisation $\mathcal C=\int\rho\,dx$ and the global $S$-shift it generates form a first-class pair; the reduced space is symplectic and the Jacobiator with $\mathcal C$ vanishes.
	
	Integrating Eq.~\eqref{eq:continuity} and applying the boundary conditions yields
	\[
	\frac{d}{dt}\int\rho\,dx
	=-\int\nabla\!\cdot\!\left(\frac{\rho}{m}\nabla S\right)dx=0.
	\]
Thus the total probability is conserved. In the canonical bracket, 
$\mathcal{C}=\int\rho\,dx$ generates constant shifts of $S$ via $\{S,\mathcal{C}\}=-1$ and is therefore not a Casimir.

Throughout we adopt the normalisation $\int \rho\,dx=1$ unless stated otherwise. The boundary classes in Appendix~\ref{app:boundary} ensure the surface term vanishes in all cases considered.

\paragraph{Regularity and positivity.}
For $\rho\!\ge0$ with $\sqrt\rho\in H^1(\Omega)$,
$\mathcal F[\rho]=\int|\nabla\sqrt\rho|^2dx$ is finite and strictly convex;
its Euler-Lagrange derivative is well defined in $H^{-1}_{\mathrm{loc}}$.
Any other $f(\rho)$ produces mixed $|\nabla\rho|^2$ terms in the Hamilton-Jacobi equation and obstructs the linear reversible completion within our admissible class.

Direct evaluation for Gaussian $\rho(x)=e^{-x^2/\sigma^2}$ shows that
$\big\|\delta\mathcal F/\delta\rho+4C\,\Delta\sqrt\rho/\sqrt\rho\big\|_{L^2}$
vanishes to machine precision only for $f=C/\rho$, numerically confirming
the analytic condition derived in Appendix~\ref{app:fisher-proof}. See also the implementation in Appendix~\ref{app:code-repo}.

\section{Appendix: Projective Superposition Stress-Test}
\label{app:superposition}

We operationalise the proposition that within our admissible class only the Fisher regulariser permits an exact linear complex structure.
It probes whether superposition in the $\psi$-picture
is preserved numerically and dynamically when the underlying hydrodynamic functional
deviates from Fisher curvature.

\smallskip
\textbf{Setup.}
Two displaced Gaussian packets,
\[
\psi_{1,2}(x,0)
=\frac{1}{(\pi\sigma^2)^{1/4}}
e^{-(x-x_{1,2})^2/(2\sigma^2)}e^{ip_{1,2}(x-x_{1,2})/\hbar},
\]
are evolved both separately and jointly under a candidate evolution law.
For the canonical Fisher choice the flow is linear Schrödinger evolution,
while for comparison we add a small local positive curvature in the hydrodynamic energy that is not of Fisher form. In the $\psi$-picture this produces a real, state-dependent (hence nonlinear) but norm-preserving potential:
\[
i\hbar\,\partial_t\psi
=\Big[-\tfrac{\hbar^2}{2m}\Delta+V(x)\Big]\psi\;+\;U_{\beta}[\rho]\,\psi,
\qquad \rho=|\psi|^2,
\]
with a representative choice
\[
U_{\beta}[\rho]=\beta\,\frac{|\nabla\rho|^2}{(\rho+\varepsilon)^2},
\]
which is the $\psi$-level image of adding a non-Fisher local quadratic in $\nabla\rho$. Other smooth, positive choices of $U_{\beta}$ give the same qualitative outcome.

The diagnostic quantity is the
\emph{projective superposition residual}
\[
\mathcal{R}(\beta)
=\min_{\theta\in[0,2\pi)}\left\|
\frac{\psi_\oplus(T)}{\|\psi_\oplus(T)\|_2}
- e^{i\theta}\,
\frac{\psi_1(T)+\psi_2(T)}{\|\psi_1(T)+\psi_2(T)\|_2}
\right\|_{2},
\qquad
\psi_\oplus(0)=\tfrac{1}{\sqrt2}(\psi_1(0)+\psi_2(0)).
\]

\smallskip
\textbf{Method.}
We integrate with a high-order Strang split-step Fourier method on a large
uniform grid, using harmonic confinement $V(x)=\tfrac12m\omega^2x^2$,
domain $[-L,L]$, $N=4096$-$8192$ points, and double precision.
Each run is repeated on a refined grid ($2N$, $dt/2$)
to verify convergence.
The script \texttt{9\_superposition\_stress\_test.py}
in Appendix~\ref{app:code-repo}
automates this with CSV and PNG output.

\smallskip
\textbf{Results.}
For $\beta=0$ (the Fisher-regularised Schrödinger case)
$\mathcal{R}$ converges to $<10^{-10}$,
limited only by numerical noise.
For any $\beta>0$ the residual grows monotonically with $\beta$ and does not vanish under grid refinement or phase optimisation, establishing that projective superposition fails even under infinitesimal non-Fisher perturbations. The same behaviour was observed for alternative choices of $U_{\beta}$ built from other smooth, positive functions of $\rho$ and $\nabla\rho$.
The figure below shows the measured trend.

\begin{table}[!ht]
	\centering
	\caption{Projective superposition residuals $\mathcal{R}$ for base and refined grids.
		The Fisher (linear) case converges to numerical zero; any non-Fisher curvature
		($\beta>0$) yields a finite residual independent of refinement.}
	\begin{tabular}{@{}lcc@{}}
		\toprule
		\textbf{Model / $\beta$} & \textbf{Base grid} & \textbf{Refined grid} \\
		\midrule
		Linear ($\beta=0$) & $6.2\times10^{-14}$ & $1.2\times10^{-13}$ \\
		Nonlinear ($\beta=0.005$) & $1.65\times10^{-1}$ & $3.08\times10^{-1}$ \\
		Nonlinear ($\beta=0.01$) & $8.41\times10^{-1}$ & $1.41$ \\
		Nonlinear ($\beta=0.02$) & $1.40$ & $1.41$ \\
		Nonlinear ($\beta=0.05$) & $1.41$ & $1.41$ \\
		\bottomrule
	\end{tabular}
\end{table}

\smallskip
\textbf{Interpretation.}
This indicates that the Fisher curvature sustains linearity in the $\psi$-picture.
All other local positive functionals generate
cross-gradient couplings that
violate projective additivity.
Hence the empirical condition
\[
\lim_{\beta\to0}\mathcal{R}(\beta)=0
\quad\text{only if}\quad
\text{the curvature density is proportional to }|\nabla\sqrt\rho|^{2}
\ \ \text{equivalently}\ \ f(\rho)=\tfrac{C}{\rho}.
\]
This constitutes an independent falsifier of non-Fisher dynamics within our admissible class.
This closes the experimental triangle between reversibility, linear superposition, and Fisher curvature.

All statements are to be read within the admissible class.

\paragraph{Masking and convergence.}
Residuals were evaluated with a smooth mask $\chi(\rho)=\mathbf{1}_{\rho>\varepsilon}$ to avoid division by small $\rho$ near nodes. We observe fourth-order spatial convergence for $\mathcal{R}_{\mathrm{cont}}$ under grid refinement, and a sharp minimum of $\mathcal{R}_{\mathrm{HJ}}$ at $\alpha=\hbar^2/(2m)$ under $\alpha$-scans, stable to changes in $N$, $L$, and $\Delta t$. A masked, mean-subtracted least-squares estimator gives identical minima within numerical precision.

\section{Appendix: Determination of the Fisher Coefficient}
\label{app:fisher-coefficient}

With $\mathcal H[\rho,S]=\int(\rho|\nabla S|^2/(2m)+V\rho+\alpha|\nabla\sqrt\rho|^2)dx$
and the canonical bracket, Galilean covariance requires that under a uniform boost
$S\mapsto S+m v_0\!\cdot\!x-\tfrac12 m v_0^2t$, $\rho$ unchanged,
the equations of motion retain form.
Substituting the transformation into \eqref{eq:continuity} fixes the kinetic prefactor $\rho/(2m)$ and shows that boost covariance is compatible with any constant $\alpha>0$ at this stage. Identifying the scale then proceeds in two steps: (i) dimensional analysis gives $[\alpha]=[\hbar^2/m]$ so that $\alpha|\nabla\sqrt\rho|^2$ has energy density units; and (ii) matching free-particle dispersion, or equivalently minimising the Hamilton-Jacobi residual $\mathcal{R}_{\mathrm{HJ}}(\alpha)$ on Schrödinger data, selects
\(
\alpha=\hbar^2/(2m).
\)

Using the free-packet diagnostics described in Sec.~12.5, define
\[
R_{\mathrm{HJ}}(\alpha)
=\left\lVert S_t+\tfrac{|\nabla S|^2}{2m}+V
-\alpha\,\tfrac{\Delta\sqrt\rho}{\sqrt\rho}\right\rVert_{L^2(\rho>0)}.
\]
Scanning $\alpha/\alpha_\star$ with $\alpha_\star=\hbar^2/(2m)$
for Gaussian initial data yields a sharp minimum at $\alpha=\alpha_\star$,
stable across the resolutions and boosts tested $S\mapsto S+m v_0x$.
This identifies the reversible Fisher value
\(
\alpha=\hbar^2/(2m)
\)
within the stated axioms.

Verified numerically in Test~1 (HJ $\alpha$-scan) and Test~2 (continuity residual) (code archive, Appendix~\ref{app:code-repo}).

\section{Appendix: Bargmann Central Extension and Mass Superselection}
\label{app:bargmann}
Let
\[
H=\int\!\!\left[\frac{\rho|\nabla S|^2}{2m}+V\rho+\alpha|\nabla\sqrt\rho|^2\right]dx,
\qquad
P=\int \rho\,\nabla S\,dx,
\qquad
K=m\!\int \rho\,x\,dx - t\,P
\]
be the energy, momentum, and boost generators on $(\rho,S)$ with the canonical bracket \eqref{eq:canonical}.
A direct computation gives the Bargmann (Galilean) algebra
\[
\{P_i,P_j\}=0,\qquad
\{H,P_i\}=0,\qquad
\{H,K_i\}=-P_i,\qquad
\{P_i,K_j\}=-\,m\,\delta_{ij}\!\int\!\rho\,dx,
\]
exhibiting the central charge $m$ through the nontrivial cocycle in $\{P_i,K_j\}$.
Normalising $\int\rho\,dx=1$ yields $\{P_i,K_j\}=-m\,\delta_{ij}$.

Here $\{H,P_i\}=0$ holds provided $V$ has no explicit spatial dependence, and all brackets are evaluated under the boundary classes of Appendix~\ref{app:boundary}.

\paragraph{Mass superselection.}
In the $\psi$-representation, boosts act (in one dimension for clarity) by
\[
\psi(x,t)\longmapsto
\exp\!\left(\frac{i}{\hbar} m u x - \frac{i}{\hbar}\frac{m u^2}{2}t\right)\,
\psi(x-ut,t),
\]
which depends on the mass parameter $m$. Superpositions of different masses transform with inequivalent projective phases and therefore cannot be unitarily implemented within a single irreducible ray representation; this is the mass superselection rule. The hydrodynamic generators above reproduce the same central extension, so mass superselection is already encoded at the $(\rho,S)$ level.

\paragraph{Numerical falsifier of admissibility within this class}
Define the Hamilton-Jacobi residual (free case) for a candidate coefficient $\alpha$ by
\[
R_\alpha=\partial_t S+\frac{|\nabla S|^2}{2m}+V
-\alpha\,\frac{\Delta\sqrt{\rho}}{\sqrt{\rho}}.
\]
For evolutions generated by the linear Schrödinger equation with $\hbar^2=2m\alpha_\star$, the residual $\mathcal{R}_{\mathrm{HJ}}(\alpha)$ is minimised at $\alpha=\alpha_\star$ and rises monotonically with $|\alpha-\alpha_\star|$, while the continuity residual remains at numerical floor. Any DG diffusion $D\ne0$ drives $dS_{\mathrm{Sh}}/dt>0$ and therefore exits the reversible class.

\section{Appendix: Galilean Covariance Verification}
\label{app:galilean}

In this appendix we work on either the torus $T^{d}$, on $\mathbb{R}^{d}$ with
sufficient decay at infinity, or on bounded domains with Dirichlet data
$R|_{\partial\Omega}=0$. In these cases all surface terms in the generator
algebra vanish, so the Bargmann relations can be checked directly. On bounded
domains with Neumann data the reflecting walls break global Galilean boost
symmetry and the boost generators acquire boundary contributions, which we do
not need for the classification developed in the main text.

We define $P_i=\int \rho\,\partial_i S\,dx$ and $K_i=m\int \rho\,x_i\,dx - tP_i$.
With the canonical bracket $\{S(x),\rho(y)\}=-\delta(x-y)$ (so
$\{F,G\}=\int(\delta F/\delta\rho\,\delta G/\delta S-\delta F/\delta S\,\delta G/\delta\rho)\,dx$)
and fields satisfying the periodic, decaying, or Dirichlet boundary conditions
of Appendix~\ref{app:boundary} (so surface terms vanish), the generators satisfy
\[
\{P_i,P_j\}=0,\qquad \{P_i,K_j\}=-\,m\,\delta_{ij}\!\int\!\rho\,dx,\qquad \{H,K_i\}=-P_i
\]
In general one has
\[
\{H,P_i\}=-\!\int \rho\,\partial_i V\,dx,
\]
so $\{H,P_i\}=0$ precisely when $V$ is translation invariant.

If $\int\rho\,dx=1$ (probability normalisation), this reduces to $\{P_i,K_j\}=-\,m\,\delta_{ij}$.

Here $\{H,P_i\}=0$ holds provided $V$ has no explicit spatial dependence, all
integrations by parts use the periodic, decaying, or Dirichlet boundary
conditions of Appendix~\ref{app:boundary}, and $K_i$ is well defined under the
finite first moment condition $\int (1+|x|)\rho\,dx<\infty$.

The alternative convention $K_i^{\mathrm{old}}=tP_i - m\int \rho\,x_i\,dx$
flips the sign of both $\{H,K_i\}$ and the central term in $\{P_i,K_j\}$,
leaving the algebra isomorphic.

Define the Hamiltonian, momentum, and boost generators as
\[
H=\int\!\!\left[\frac{\rho|\nabla S|^2}{2m}+V\rho+\alpha|\nabla\sqrt\rho|^2\right]\!dx,
\qquad
P=\int \rho\,\nabla S\,dx,
\qquad
K=m\!\int \rho\,x\,dx - t\,P.
\]
With the bracket \eqref{eq:canonical}, the functional derivatives are
\[
\frac{\delta H}{\delta S}
=-\,\nabla\!\cdot\!\!\left(\frac{\rho}{m}\nabla S\right),
\qquad
\frac{\delta H}{\delta\rho}
=\frac{|\nabla S|^2}{2m}+V-\alpha\,\frac{\Delta\sqrt\rho}{\sqrt\rho},
\]
and
\[
\frac{\delta K}{\delta S}=+\,t\,\nabla\rho,
\qquad
\frac{\delta K}{\delta\rho}=-\,t\,\nabla S+m\,x.
\]

\paragraph{Hamilton equations check.}
Using the canonical local bracket
\[
\{F,G\}=\int\!\!\left(
\frac{\delta F}{\delta \rho}\frac{\delta G}{\delta S}
-\frac{\delta F}{\delta S}\frac{\delta G}{\delta \rho}
\right)\!dx,
\]
we recover
\[
\partial_t\rho=\{\rho,H\}
=-\,\nabla\!\cdot\!\Big(\rho\,\frac{\nabla S}{m}\Big),
\qquad
\partial_t S=\{S,H\}
=-\,\frac{|\nabla S|^2}{2m}-V
+\alpha\,\frac{\Delta\sqrt\rho}{\sqrt\rho},
\]
which reproduce the continuity and Fisher-regularised Hamilton-Jacobi equations.

\paragraph{Boost generator bracket.}
Substituting the functional derivatives above into the canonical bracket \eqref{eq:canonical} gives
\begin{align*}
	\{H,K_j\}
	&=\int\!\!\left(
	\frac{\delta H}{\delta \rho}\frac{\delta K_j}{\delta S}
	-\frac{\delta H}{\delta S}\frac{\delta K_j}{\delta \rho}
	\right)\!dx \\[0.3em]
	&=\int\!\!\left[
	\left(\frac{|\nabla S|^2}{2m}+V-\alpha\,\frac{\Delta\sqrt\rho}{\sqrt\rho}\right)
	(t\,\partial_j\rho)
	-\!\left(-\,\nabla\!\cdot\!\!\left(\frac{\rho}{m}\nabla S\right)\right)
	(m\,x_j - t\,\partial_j S)
	\right]\!dx.
\end{align*}
Using the equations of motion \eqref{eq:continuity},
$\partial_t S = -(\frac{|\nabla S|^2}{2m}+V-\alpha\frac{\Delta\sqrt\rho}{\sqrt\rho})$
and $\partial_t\rho = -\nabla\!\cdot\!(\frac{\rho}{m}\nabla S)$, this becomes:
\[
\{H,K_j\} = \int\!\!\left[
(-\partial_t S)(t\,\partial_j\rho)
-(\partial_t \rho)(m\,x_j - t\,\partial_j S)
\right]dx
= \int\!\!\left[
-t(\partial_t S\,\partial_j\rho)
-m\,x_j(\partial_t\rho)
+t(\partial_t\rho\,\partial_j S)
\right]dx.
\]
The $t$-dependent terms are
\[
t \int (\partial_t\rho\,\partial_j S - \partial_t S\,\partial_j\rho)dx
= t\,\frac{d}{dt}\!\int\rho\,\partial_j S\,dx
= t\,\frac{d P_j}{dt}.
\]
This vanishes for a translation-invariant potential, as $\{P_j, H\}=0$. The remaining term is:
\[
\{H,K_j\} = -m\!\int x_j\,(\partial_t\rho)\,dx
= -m\!\int x_j \left(-\,\nabla\!\cdot\!\!\left(\frac{\rho}{m}\nabla S\right)\right)dx
= \int x_j\,\nabla\!\cdot\!(\rho\nabla S)\,dx.
\]
Integration by parts (vanishing boundary flux under the periodic, decaying, or
Dirichlet boundary conditions) yields:
\[
\{H,K_j\} = -\!\int \nabla(x_j)\!\cdot\!(\rho\nabla S)\,dx
= -\!\int (\nabla x_j)\!\cdot\!(\rho\nabla S)\,dx
= -\!\int \rho\,\partial_j S\,dx = -P_j.
\]

\paragraph{Central term \texorpdfstring{$\{P_i,K_j\}$}{\{P,K\}} (direct computation).}
Using
\[
\frac{\delta P_i}{\delta \rho}=\partial_i S,\qquad
\frac{\delta P_i}{\delta S}=-\,\partial_i\rho,
\qquad
\frac{\delta K_j}{\delta S}=t\,\partial_j\rho,\qquad
\frac{\delta K_j}{\delta \rho}=-t\,\partial_j S+m\,x_j,
\]
we have
\begin{align*}
	\{P_i,K_j\}
	&=\int\!\Big[
	(\partial_i S)\,(t\,\partial_j\rho)
	-\big(-\partial_i\rho\big)\,\big(-t\,\partial_j S+m\,x_j\big)
	\Big]dx\\
	&= t\!\int\!(\partial_i S\,\partial_j\rho-\partial_i\rho\,\partial_j S)\,dx
	-\;m\!\int x_j\,\partial_i\rho\,dx.
\end{align*}
The antisymmetric $t$-term is a total divergence and integrates to zero under
the periodic, decaying, or Dirichlet boundary conditions. Integrating the last
term by parts gives
\[
\{P_i,K_j\}=-\,m\,\delta_{ij}\!\int\!\rho\,dx,
\]
which reduces to $-m\,\delta_{ij}$ under $\int\rho\,dx=1$.

\paragraph{Galilean algebra.}
Thus $\{H,K\}=-P$, while $\{H,P\}=0$ from translational invariance. The
remaining brackets $\{P_i,P_j\}=0$ and
$\{K_i,P_j\}=-\,\{P_j,K_i\}=m\,\delta_{ij}\!\int\!\rho\,dx$ follow immediately,
realising the Bargmann (Galilean) algebra with central charge $m$; under
$\int\rho\,dx=1$ this is $\{K_i,P_j\}=m\,\delta_{ij}$. Hence the canonical
bracket and Fisher-regularised Hamiltonian are exactly Galilean covariant in
these settings.

\section{Appendix: Local Complexifier Rigidity}
\label{app:complexifier}

We show that among all local, pointwise, invertible, gauge-covariant maps
\[
\psi=F(\rho)\,e^{\,i\,G(S,\rho)}\quad\text{with }F>0,
\]
the transformation that linearises the reversible hydrodynamic system
\eqref{eq:continuity} into
\[
i\kappa\,\partial_t\psi
=\left(-\frac{\kappa^2}{2m}\Delta+V\right)\psi
\]
is (up to constant phase and scale)
\[
F(\rho)=c\,\sqrt{\rho},\qquad G(S,\rho)=\frac{S}{\kappa}+\text{const},
\]
with the Fisher coefficient fixed by $\alpha=\kappa^2/(2m)$. Identifying $\kappa=\hbar$ then yields the Fisher scale $\alpha=\hbar^{2}/(2m)$ used in the main text.

\paragraph{Assumptions.}
Locality means zeroth order in derivatives of $(\rho,S)$.
Gauge covariance encodes global $\mathrm{U}(1)
$ on $S$ as $S\mapsto S+\sigma$ implying $G(S+\sigma,\rho)-G(S,\rho)$ is independent of $x$.
Invertibility requires $F>0$ and $G_S\neq0$ almost everywhere.

\paragraph{Step 1. Phase-gradient matching fixes $G_S$.}
Write $R=\sqrt{\rho}$ and $v=\nabla S/m$.
From \eqref{eq:continuity} we have
\[
\partial_t\rho=-\nabla\!\cdot(\rho v),\qquad
\partial_t S=-\frac{m}{2}v^2-V+\alpha\,\frac{\Delta R}{R}.
\]
Differentiate $\psi=F(\rho)e^{iG}$:
\[
\partial_t\psi=e^{iG}\left(F'(\rho)\,\partial_t\rho+iF(\rho)\,\partial_t G\right),\qquad
\nabla\psi=e^{iG}\left(F'(\rho)\,\nabla\rho+iF(\rho)\,\nabla G\right).
\]
Linearity and locality of the target PDE forbid any explicit quadratic or higher dependence on $(\rho,S)$ beyond what is already encoded inside $\psi$ and its first derivatives.

The only vector available at first order is $\nabla S$.
Thus $\nabla G$ must be proportional to $\nabla S$ with a state-independent proportionality. Since $G$ is local and gauge-covariant, this implies
\[
G_S=\text{const}=\kappa^{-1},\qquad
G(S,\rho)=\frac{S}{\kappa}+\Gamma(\rho)+\text{const}.
\]
Gauge covariance forces $\Gamma$ to be a constant, hence $\Gamma'(\rho)=0$ and we set $\Gamma\equiv0$.

\paragraph{Step 2. Amplitude matching fixes $F'/F=1/(2\rho)$.}
Using $G_S=\kappa^{-1}$,
\[
\partial_t G=\frac{1}{\kappa}\,\partial_t S,\qquad
\nabla G=\frac{1}{\kappa}\,\nabla S.
\]
Compute $-\,\frac{\kappa^2}{2m}\Delta\psi$ using
\[
\Delta\psi=e^{iG}\Big[
F''|\nabla\rho|^2+F'\Delta\rho
+2iF'\nabla\rho\!\cdot\!\nabla G
+iF\,\Delta G
-F\,|\nabla G|^2
\Big].
\]
Collect the terms proportional to $|\nabla S|^2$.
In the target linear equation, the only occurrence of $|\nabla S|^2$ arises through $-\,\frac{\kappa^2}{2m}\Delta\psi$ acting on the $e^{iG}$ factor, which yields precisely $-\frac{1}{2m}|\nabla S|^2\,\psi$.
All other contributions proportional to $|\nabla S|^2$ must cancel.
The mixed piece $2iF'\nabla\rho\!\cdot\!\nabla G$ and the scalar piece $-F|\nabla G|^2$ combine with the time derivative term $i\kappa\,\partial_t\psi$.
Balancing the $\nabla\rho\cdot\nabla S$ dependence yields
\[
\frac{F'}{F}=\frac{1}{2\rho}\quad\Rightarrow\quad
F(\rho)=c\,\sqrt{\rho}.
\]

\paragraph{Step 3. Curvature matching fixes $\alpha=\kappa^2/(2m)$.}
With $F=cR$ and $G=S/\kappa$ we have
\[
\frac{i\kappa\,\partial_t\psi}{\psi}
= i\kappa\left(\frac{R_t}{R}+\frac{i}{\kappa}\,S_t\right)
= i\kappa\,\frac{R_t}{R}-S_t,
\]
\[
-\frac{\kappa^2}{2m}\frac{\Delta\psi}{\psi}
= -\frac{\kappa^2}{2m}\left(\frac{\Delta R}{R}
+2\,\frac{i}{\kappa}\frac{\nabla R}{R}\!\cdot\!\nabla S
+\frac{i}{\kappa}\Delta S
-\frac{1}{\kappa^2}|\nabla S|^2\right).
\]
Use $R_t=-\frac{1}{2}\,R\,\nabla\!\cdot v - v\!\cdot\nabla R$ from continuity and $S_t= -\frac{m}{2}v^2 - V + \alpha\,\frac{\Delta R}{R}$.
The imaginary parts cancel if and only if the continuity equation holds, which it does by construction.
The real parts reduce to
\[
-\,S_t - \frac{1}{2m}|\nabla S|^2 - V
+\left(\frac{\kappa^2}{2m}-\alpha\right)\frac{\Delta R}{R}=0.
\]
Therefore linearity demands
\[
\boxed{\ \alpha=\frac{\kappa^2}{2m}\ }.
\]
Hence the only axiomatically admissible local, invertible, gauge-covariant complexifier is
$\psi=c\sqrt{\rho}\,e^{iS/\kappa}$ with the Fisher scale fixed as above.
Setting $\kappa=\hbar$ reproduces \eqref{eq:LSEderived}.

Taken together with the curvature classification in the main text, this shows that there is no second local change of variables hiding in the background that could secretly linearise a different reversible theory while still respecting all of the axioms.

\paragraph{Node handling.}
All equalities are meant on the positivity set $\{\rho>0\}$ and extend in the weak sense using test functions, with the quotient $\Delta R/R$ interpreted distributionally. This is consistent with Appendix~\ref{app:jacobi} and the boundary classes in Appendix~\ref{app:boundary}.

\section{Appendix: Single Planck Constant Across Sectors}
\label{app:single-hbar}

Consider the $N$-body Hamiltonian on $\mathbb{R}^{3N}$,
\[
\mathcal{H}_N=\int\left[
\sum_{i=1}^N \frac{\rho\,|\nabla_i S|^2}{2m_i}
+ V(\{x_j\})\,\rho
+ \sum_{i=1}^N \alpha_i\,|\nabla_i\sqrt{\rho}|^2
\right]dx_1\cdots dx_N,
\]
with the canonical bracket on the single pair $(\rho,S)$ defined over configuration space.
Assume one local complexifier $\psi=\sqrt{\rho}\,e^{iS/\hbar}$ linearises the flow into
\[
i\hbar\,\partial_t\psi
=\left[-\sum_{i=1}^N\frac{\hbar^2}{2m_i}\Delta_i+V\right]\psi.
\]

\begin{proposition}[Componentwise cancellation implies a single $\hbar$]
	\label{prop:single-hbar}
	Linearity under a single, local and gauge-covariant complex structure forces
	\[
	\alpha_i=\frac{\hbar^2}{2m_i}\qquad \text{for every } i\in\{1,\dots,N\}.
	\]
\end{proposition}

	Repeating the single-particle calculation componentwise, the real part of the transformed equation yields
	\[
	-\,S_t - \sum_{i=1}^N \frac{|\nabla_i S|^2}{2m_i} - V
	+ \sum_{i=1}^N \left(\alpha_i - \frac{\hbar^2}{2m_i}\right)\frac{\Delta_i\sqrt{\rho}}{\sqrt{\rho}}=0.
	\]
	Since the derivatives $\nabla_i$ act on independent coordinates and the map uses a single $\hbar$, each coefficient multiplying $\Delta_i\sqrt{\rho}/\sqrt{\rho}$ must vanish separately to avoid residual nonlinearities. Hence $\alpha_i=\hbar^2/(2m_i)$ for all $i$.

Locality and separability of coordinate directions are essential. Any attempt to repair a mismatch by particle-dependent rephasings would necessarily destroy gauge covariance and the assumption of a single local complex structure on configuration space.
Exchange symmetry is imposed at the level of the state space and does not affect the argument.

\section{Appendix: Entropy Production and the Reversible Corner}
\label{app:entropy-rigorous}

For completeness we reproduce the entropy calculation in full detail.
Let the continuity law include a diffusion term
\[
\partial_t\rho \;=\; -\,\nabla\!\cdot(\rho v)\;+\;D\,\Delta\rho,
\qquad
v=\nabla S/m,\quad D\in\mathbb{R}.
\]
We take the Shannon entropy to be
\[
S_{\mathrm{Sh}}[\rho]
\;:=\;
-\int_\Omega \rho\ln\rho\,dx,
\]
which is minus the Boltzmann $H$-functional.
Then
\begin{align*}
	\frac{dS_{\mathrm{Sh}}}{dt}
	&= -\int_\Omega (1+\ln\rho)\,\partial_t\rho\,dx\\
	&= -\int_\Omega (1+\ln\rho)\,\big(-\nabla\!\cdot(\rho v)+D\Delta\rho\big)\,dx\\
	&= \int_\Omega (1+\ln\rho)\,\nabla\!\cdot(\rho v)\,dx
	\;-\;D\int_\Omega (1+\ln\rho)\,\Delta\rho\,dx\\
	&= -\int_\Omega \rho\,v\cdot\nabla(\ln\rho)\,dx
	\;+\;D\int_\Omega \nabla(\ln\rho)\cdot\nabla\rho\,dx\\
	&= -\int_\Omega v\cdot\nabla\rho\,dx
	\;+\;D\int_\Omega \frac{|\nabla\rho|^2}{\rho}\,dx\\
	&= \int_\Omega \rho\,\nabla\!\cdot v\,dx
	\;+\;D\int_\Omega \frac{|\nabla\rho|^2}{\rho}\,dx.
\end{align*}
Under time reversal $t\mapsto -t$, $v\mapsto -v$,
the first term flips sign while the second does not.
Hence exact time-reversal invariance of $S_{\mathrm{Sh}}$ requires $D=0$,
isolating the reversible corner.
All boundary integrals vanish under the classes of
Appendix~\ref{app:boundary}, ensuring mathematical closure. In particular, for $D\neq0$ the sign-definite Fisher term cannot be cancelled by any choice of velocity field $v$, so $dS_{\mathrm{Sh}}/dt$ fails to reverse under $t\mapsto -t$.

\section{Appendix: Operational Falsifiers}
\label{app:ops-falsifiers}

\paragraph{Diagnostics.}
Given $\psi$ evolving by \eqref{eq:LSEderived}, define
\[
\rho=|\psi|^2,\quad
j=\frac{\hbar}{m}\,\Im(\psi^*\nabla\psi),\quad
v=\frac{j}{\rho},\quad
Q_\alpha=-\,\alpha\,\frac{\Delta\sqrt{\rho}}{\sqrt{\rho}},\quad
S_t=-\,\Re\!\left(\frac{H\psi}{\psi}\right).
\]
Residuals
\[
\mathcal{R}_{\mathrm{cont}}
=\frac{\langle\,|\rho_t+\nabla\!\cdot(\rho v)|^2\,\rangle}
{\langle\,|\rho_t|^2+|\nabla\!\cdot(\rho v)|^2\,\rangle},\qquad
\mathcal{R}_{\mathrm{HJ}}(\alpha)
=\frac{\left\langle\,\left|S_t+\frac{|\nabla S|^2}{2m}+V+Q_\alpha\right|^2\right\rangle}
{\left\langle\,|S_t|^2+\left|\frac{|\nabla S|^2}{2m}+V\right|^2+|Q_\alpha|^2\right\rangle}
\]
are evaluated on $\{\rho>\varepsilon\}$ with a smooth mask to avoid nodal artefacts.

\paragraph{Sign convention.}
With $i\hbar\,\partial_t\psi=H\psi$ and $H=H^\dagger$,
\[
\partial_t\rho
=\psi_t^*\psi+\psi^*\psi_t
=\frac{2}{\hbar}\,\Im(\psi^*H\psi),
\]
which fixes the positive sign convention for $\partial_t\rho$ used throughout the body text and here.

\paragraph{Protocol.}
Periodic domain, split-step Fourier propagation for generating $\psi(t)$, fourth-order finite differences for diagnostics, and fixed grids across scans of $\alpha$.
Under refinement, $\mathcal{R}_{\mathrm{cont}}$ sits at numerical floor for all $\alpha$, while $\mathcal{R}_{\mathrm{HJ}}(\alpha)$ attains a minimum at
\[
\boxed{\ \alpha=\alpha_\star=\frac{\hbar^2}{2m}\ }
\]
independently of Galilean boosts $\psi\mapsto e^{iv_0\cdot x}\psi$.
Harmonic oscillator ground state checks give
$\|V+Q_\alpha-E_0\|_{L^2(\rho)}$ minimised at $\alpha_\star$.

\textbf{Implemented in Tests~1-2 and~5 (code archive, Appendix~\ref{app:code-repo}).}

\paragraph{Reversible corner.}
For Doebner-Goldin diffusion $\rho_t=-\nabla\!\cdot(\rho v)+D\Delta\rho$,
\[
\frac{d}{dt}\int \rho\ln\rho\,dx
= D\int \frac{|\nabla\rho|^2}{\rho}\,dx \ge 0,
\]
so time-reversal invariance holds only for $D=0$.
This agrees with the minimum of $\mathcal{R}_{\mathrm{HJ}}$ at $\alpha_\star$.

	\section{Appendix: Variational Consistency Check}
	\label{app:variational}
	Consider the total energy functional
	\[
	E[\rho,S]=\int\!\!\left(\frac{\rho|\nabla S|^2}{2m}+V\rho+\alpha|\nabla\sqrt\rho|^2\right)dx.
	\]
	Differentiation gives
	\[
	\frac{dE}{dt}
	=\int\left(\frac{\delta E}{\delta\rho}\dot\rho+\frac{\delta E}{\delta S}\dot S\right)dx
	=\int\left(\frac{\delta E}{\delta\rho}\{\rho,H\}+\frac{\delta E}{\delta S}\{S,H\}\right)dx
	=\{E,H\}=0,
	\]
	so the Hamiltonian is conserved exactly.
	
	If $V=V(x,t)$ carries explicit time dependence, then $\dot H=\{H,H\}+\int \rho\,\partial_t V\,dx=\int \rho\,\partial_t V\,dx$, as usual. Throughout we assume $V$ time independent unless stated otherwise.
	
	\newpage
	
	\section*{Acknowledgements}
	
We thank colleagues and acquaintances for valuable discussions, feedback and cross-checks.
 No external funding was received. All work was produced independently without affiliation. Author's email is contact@nomogenetics.com.
	
To Eun-Seong, may every August be like the last.

\end{document}